\documentclass{IEEEtran}
\ifCLASSINFOpdf
\else
\fi
\usepackage{url}

\usepackage[english]{babel}
\usepackage[autostyle, english=american]{csquotes}
\MakeOuterQuote{"}

\usepackage{amsbsy}
\usepackage{amsmath}
\usepackage{array}
\usepackage{fixmath}
\usepackage{amssymb}
\usepackage{epsfig}
\usepackage{epstopdf}
\usepackage{algorithmic}
\usepackage{algorithm}
\usepackage{amsthm}
\usepackage{multirow}

\newtheorem{theorem}{Theorem}

\newtheorem{lemma}[theorem]{Lemma}

\usepackage{tikz}
\usetikzlibrary{arrows,decorations,backgrounds,shadows,plotmarks,calc, positioning}
\usetikzlibrary{arrows.meta}

\usepackage{pgfplots}
\pgfplotsset{compat=newest}
\pgfplotsset{plot coordinates/math parser=false}
\pgfplotsset{every axis/.append style={font=\footnotesize}}
\pgfplotsset{
    ylabel right/.style={
        after end axis/.append code={
            \node [rotate=90, anchor=north] at (rel axis cs:1,0.5) {#1};
        }   
    }
}
\newlength\figureheight
\newlength\figurewidth
\newlength\subgraphheight
\newlength\subgraphwidth
\makeatletter
\newcommand{\nodedistance}{\tikz@node@distance}
\makeatother

\makeatletter

\pgfkeys{
	/tikz/blockm/.initial=1,
	/tikz/blockn/.initial=1
}

\pgfdeclareshape{dspantenna}{
	
	\inheritsavedanchors[from=rectangle]
	
	\inheritanchorborder[from=rectangle]
	
	\inheritanchor[from=rectangle]{center}
	\inheritanchor[from=rectangle]{south}
	\inheritanchor[from=rectangle]{west}
	\inheritanchor[from=rectangle]{north}
	\inheritanchor[from=rectangle]{east}
	\inheritanchor[from=rectangle]{south east}
	\inheritanchor[from=rectangle]{south west}
	\inheritanchor[from=rectangle]{north east}
	\inheritanchor[from=rectangle]{north west}
	\saveddimen{\halfheight}{
		\pgfmathsetlength{\pgf@xa}{\pgfkeysvalueof{/pgf/minimum height}/2}
		\pgfmathsetlength{\pgf@xb}{\pgfkeysvalueof{/pgf/outer ysep}}
		\advance\pgf@xa by \pgf@xb
		\pgf@x=\pgf@xa
	}
	\saveddimen{\halfwidth}{
		\pgfmathsetlength{\pgf@xa}{\pgfkeysvalueof{/pgf/minimum width}/2}
		\pgfmathsetlength{\pgf@xb}{\pgfkeysvalueof{/pgf/outer xsep}}
		\advance\pgf@xa by \pgf@xb
		\pgf@x=\pgf@xa
	}
	
	\backgroundpath{
		
		\northeast \pgf@xa=\pgf@x \pgf@ya=\pgf@y
		\pgf@xb=\pgf@xa \pgf@yb=\pgf@ya
		
		\pgfpathmoveto{\pgfpoint{\pgf@xb}{\pgf@yb}}
		\advance \pgf@xb by -\halfwidth
		\advance \pgf@yb by -\halfwidth
		\pgfpathlineto{\pgfpoint{\pgf@xb}{\pgf@yb}}
		\advance \pgf@xb by -\halfwidth
		\advance \pgf@yb by \halfwidth
		\pgfpathlineto{\pgfpoint{\pgf@xb}{\pgf@yb}}
		\advance \pgf@xb by \halfwidth
		\advance \pgf@yb by -\halfwidth
		\pgfpathmoveto{\pgfpoint{\pgf@xb}{\pgf@yb}}
		\advance \pgf@yb by -\halfwidth
		\pgfpathlineto{\pgfpoint{\pgf@xb}{\pgf@yb}}
		\advance \pgf@xb by \halfwidth
		\pgfpathlineto{\pgfpoint{\pgf@xb}{\pgf@yb}}
		
		\pgfusepath{draw}
	}
}

\pgfdeclareshape{dspvdots}{
	
	\inheritsavedanchors[from=rectangle]
	
	\inheritanchorborder[from=rectangle]
	
	\inheritanchor[from=rectangle]{center}
	\inheritanchor[from=rectangle]{south}
	\inheritanchor[from=rectangle]{west}
	\inheritanchor[from=rectangle]{north}
	\inheritanchor[from=rectangle]{east}
	\inheritanchor[from=rectangle]{south east}
	\inheritanchor[from=rectangle]{south west}
	\inheritanchor[from=rectangle]{north east}
	\inheritanchor[from=rectangle]{north west}
	\saveddimen{\halfheight}{
		\pgfmathsetlength{\pgf@xa}{\pgfkeysvalueof{/pgf/minimum height}/2}
		\pgfmathsetlength{\pgf@xb}{\pgfkeysvalueof{/pgf/outer ysep}}
		\advance\pgf@xa by \pgf@xb
		\pgf@x=\pgf@xa
	}
	\saveddimen{\halfwidth}{
		\pgfmathsetlength{\pgf@xa}{\pgfkeysvalueof{/pgf/minimum width}/2}
		\pgfmathsetlength{\pgf@xb}{\pgfkeysvalueof{/pgf/outer xsep}}
		\advance\pgf@xa by \pgf@xb
		\pgf@x=\pgf@xa
	}
	
	\backgroundpath{
		
		\northeast \pgfutil@tempdima=\pgf@x \pgfutil@tempdimb=\pgf@y
		
		\pgf@xa=\halfwidth
		\pgf@ya=\halfheight
		\advance\pgfutil@tempdima by -\pgf@xa
		\advance\pgfutil@tempdimb by -\pgf@ya 
		
		\pgfpathcircle{\pgfpoint{\pgfutil@tempdima}{\pgfutil@tempdimb}}{\halfwidth/3}
		\pgf@xa=\halfwidth
		\pgf@ya=\halfheight
		\advance\pgfutil@tempdimb by \pgf@xa
		\advance\pgfutil@tempdimb by -\pgf@ya
		\pgfpathcircle{\pgfpoint{\pgfutil@tempdima}{\pgfutil@tempdimb}}{\halfwidth/3}
		\pgf@xa=\halfwidth
		\pgf@ya=\halfheight
		\advance\pgfutil@tempdimb by 2\pgf@ya
		\advance\pgfutil@tempdimb by -2\pgf@xa
		\pgfpathcircle{\pgfpoint{\pgfutil@tempdima}{\pgfutil@tempdimb}}{\halfwidth/3} 
		
		\pgfusepath{fill}
	}
}

\pgfdeclareshape{dspadder}{
	\inheritsavedanchors[from=circle]
	\inheritanchorborder[from=circle]
	
	\inheritanchor[from=circle]{center}
	\inheritanchor[from=circle]{south}
	\inheritanchor[from=circle]{west}
	\inheritanchor[from=circle]{north}
	\inheritanchor[from=circle]{east}
	\inheritanchor[from=circle]{south east}
	\inheritanchor[from=circle]{south west}
	\inheritanchor[from=circle]{north east}
	\inheritanchor[from=circle]{north west}
	
	\inheritbackgroundpath[from=circle]
	
	\beforebackgroundpath{
		
		\radius \pgf@xb=\pgf@x
		\centerpoint \pgf@xa=\pgf@x \pgf@ya=\pgf@y \pgf@xc=\pgf@x \pgf@yc=\pgf@y
		
		\advance\pgf@yc by \pgf@xb
		\pgfpathmoveto{\pgfpoint{\pgf@xc}{\pgf@yc}}
		\advance\pgf@yc by -2\pgf@xb
		\pgfpathlineto{\pgfpoint{\pgf@xc}{\pgf@yc}}
		\advance\pgf@xc by -\pgf@xb
		\advance\pgf@yc by \pgf@xb
		\pgfpathmoveto{\pgfpoint{\pgf@xc}{\pgf@yc}}
		\advance\pgf@xc by 2\pgf@xb
		\pgfpathlineto{\pgfpoint{\pgf@xc}{\pgf@yc}}
		
		\pgfusepath{draw}
	}
}

\pgfdeclareshape{dsptriangler}{
	
	\inheritsavedanchors[from=rectangle]
	
	\inheritanchorborder[from=rectangle]
	
	\inheritanchor[from=rectangle]{center}
	\inheritanchor[from=rectangle]{south}
	\inheritanchor[from=rectangle]{west}
	\inheritanchor[from=rectangle]{north}
	\inheritanchor[from=rectangle]{east}
	\inheritanchor[from=rectangle]{south east}
	\inheritanchor[from=rectangle]{south west}
	\inheritanchor[from=rectangle]{north east}
	\inheritanchor[from=rectangle]{north west}
	\saveddimen{\halfheight}{
		\pgfmathsetlength{\pgf@xa}{\pgfkeysvalueof{/pgf/minimum height}/2}
		\pgfmathsetlength{\pgf@xb}{\pgfkeysvalueof{/pgf/outer ysep}}
		\advance\pgf@xa by \pgf@xb
		\pgf@x=\pgf@xa
	}
	\saveddimen{\halfwidth}{
		\pgfmathsetlength{\pgf@xa}{\pgfkeysvalueof{/pgf/minimum width}/2}
		\pgfmathsetlength{\pgf@xb}{\pgfkeysvalueof{/pgf/outer xsep}}
		\advance\pgf@xa by \pgf@xb
		\pgf@x=\pgf@xa
	}
	
	\backgroundpath{

		\pgfutil@tempdima=\halfheight
		\pgfutil@tempdimb=\halfwidth
		\southwest \pgf@xa=\pgf@x \pgf@ya=\pgf@y 
		\pgfpathmoveto{\pgfpoint{\pgf@xa}{\pgf@ya}}
		\advance \pgf@ya by 2\pgfutil@tempdima
		\pgfpathlineto{\pgfpoint{\pgf@xa}{\pgf@ya}}
		\advance \pgf@xa by 2\pgfutil@tempdimb
		\advance \pgf@ya by -\pgfutil@tempdima
		\pgfpathlineto{\pgfpoint{\pgf@xa}{\pgf@ya}}	
		\advance \pgf@xa by -2\pgfutil@tempdimb
		\advance \pgf@ya by -\pgfutil@tempdima
		\pgfpathlineto{\pgfpoint{\pgf@xa}{\pgf@ya}}       
		\pgfusepath{draw}
	}
}

\pgfdeclareshape{dspdiamond}{
	
	\inheritsavedanchors[from=rectangle]
	
	\inheritanchorborder[from=rectangle]
	
	\inheritanchor[from=rectangle]{center}
	\inheritanchor[from=rectangle]{south}
	\inheritanchor[from=rectangle]{west}
	\inheritanchor[from=rectangle]{north}
	\inheritanchor[from=rectangle]{east}
	\inheritanchor[from=rectangle]{south east}
	\inheritanchor[from=rectangle]{south west}
	\inheritanchor[from=rectangle]{north east}
	\inheritanchor[from=rectangle]{north west}
	\saveddimen{\halfheight}{
		\pgfmathsetlength{\pgf@xa}{\pgfkeysvalueof{/pgf/minimum height}/2}
		\pgfmathsetlength{\pgf@xb}{\pgfkeysvalueof{/pgf/outer ysep}}
		\advance\pgf@xa by \pgf@xb
		\pgf@x=\pgf@xa
	}
	\saveddimen{\halfwidth}{
		\pgfmathsetlength{\pgf@xa}{\pgfkeysvalueof{/pgf/minimum width}/2}
		\pgfmathsetlength{\pgf@xb}{\pgfkeysvalueof{/pgf/outer xsep}}
		\advance\pgf@xa by \pgf@xb
		\pgf@x=\pgf@xa
	}
	
	\backgroundpath{

		\pgfutil@tempdima=\halfheight
		\pgfutil@tempdimb=\halfwidth
		\southwest \pgf@xa=\pgf@x \pgf@ya=\pgf@y 
		\advance\pgf@xa by \pgfutil@tempdimb
		\pgfpathmoveto{\pgfpoint{\pgf@xa}{\pgf@ya}}
		\advance \pgf@xa by \pgfutil@tempdimb
		\advance \pgf@ya by \pgfutil@tempdima
		\pgfpathlineto{\pgfpoint{\pgf@xa}{\pgf@ya}}
		\advance \pgf@xa by -\pgfutil@tempdimb
		\advance \pgf@ya by \pgfutil@tempdima
		\pgfpathlineto{\pgfpoint{\pgf@xa}{\pgf@ya}}	     	
		\advance \pgf@xa by -\pgfutil@tempdimb
		\advance \pgf@ya by -\pgfutil@tempdima
		\pgfpathlineto{\pgfpoint{\pgf@xa}{\pgf@ya}}      	
		\advance \pgf@xa by \pgfutil@tempdimb
		\advance \pgf@ya by -\pgfutil@tempdima
		\pgfpathlineto{\pgfpoint{\pgf@xa}{\pgf@ya}}     
		\pgfusepath{draw}
	}
}

\makeatother

\usepackage{pgfplots}
\usepackage{pgfkeys}
\pgfplotsset{compat=newest}
\pgfplotsset{plot coordinates/math parser=false}

\DeclareMathOperator{\sign}{sign}

\DeclareMathOperator{\sinc}{sinc}

\DeclareMathAlphabet{\mathbit}{OML}{cmr}{bx}{it}
\DeclareMathAlphabet{\mathsf}{OT1}{cmss}{m}{n}
\DeclareMathAlphabet{\mathbsf}{OT1}{cmss}{bx}{it}
\newcommand{\inC}[1]{\ensuremath{\in\mathbb{C}^{#1}}}

\newcommand{\inset}[2]{\ensuremath{\in \left\{#1,\ldots,#2\right\}}}

\newcommand{\Imag}[1]{\ensuremath{\Im\left\{#1\right\}}}

\newcommand\diff[1]{\ensuremath{\:\mathrm{d}#1}}

\newcommand\derivk[3]{\ensuremath{\frac{\mathrm{d}^{#3}#1}{\mathrm{d}{#2}^{#3}}}}
\newcommand\pderiv[2]{\ensuremath{\frac{\partial#1}{\partial#2}}}
\newcommand\pderivk[3]{\ensuremath{\frac{\partial^{#3}#1}{\partial{#2}^{#3}}}}

\usepackage{etoolbox}
\makeatletter
\patchcmd{\@maketitle}
{\addvspace{0.5\baselineskip}\egroup}
{\addvspace{-0.5\baselineskip}\egroup}
{}
{}

\def\ps@IEEEtitlepagestyle{%
	\def\@oddfoot{\mycopyrightnotice}%
	\def\@evenfoot{}%
}
\def\mycopyrightnotice{%
	{\begin{minipage}{\textwidth}\centering\footnotesize This work has been accepted for publication in IEEE/OSA Journal of Lightwave Technology. Digital Object Identifier: 10.1109/JLT.2018.2875557 \\ Copyright \textcopyright 2018 IEEE. Personal use of this material is permitted. Permission from IEEE must be obtained for all other uses, in any current or future media, including reprinting/republishing this material for advertising or promotional purposes, creating new collective works, for resale or redistribution to servers or lists, or reuse of any copyrighted component of this work in other works.	\end{minipage}}
	\gdef\mycopyrightnotice{}
}

\makeatother

\begin{document}
%
\title{Communication Using Eigenvalues of Higher Multiplicity of the Nonlinear Fourier Transform}

\author{Javier Garc\'ia
\IEEEcompsocitemizethanks{
	\IEEEcompsocthanksitem Date of current version \today. J. Garc\'ia is with the Institute for Communications Engineering (LNT), Technical University of Munich, Munich 80333, Germany (e-mail: javier.garcia@tum.de). His work was supported by the German Research Foundation under Grant KR 3517/8-1. 
	}
}
\maketitle

\begin{abstract}
A generalized Nonlinear Fourier Transform (GNFT), which includes eigenvalues of higher multiplicity, is considered for information transmission over fiber optic channels. Numerical algorithms are developed to compute the direct and inverse GNFTs. For closely-spaced eigenvalues, examples suggest that the GNFT is more robust than the NFT to the practical impairments of truncation, discretization, attenuation and noise. Communication using a soliton with one double eigenvalue is numerically demonstrated, and its information rates are compared to solitons with one and two simple eigenvalues. 
\end{abstract}

\begin{IEEEkeywords}
	Inverse Scattering Transform, Nonlinear Fourier Transform, optical fiber, higher multiplicity eigenvalues, spectral efficiency
\end{IEEEkeywords}

%
\IEEEpeerreviewmaketitle

\section{Introduction}\label{sec:intro}
Current optical transmission systems exhibit a peak in the achievable rate due to the Kerr nonlinearity of the Nonlinear Schr\"odinger Equation (NLSE)~\cite{essiambre_limits}. Several techniques have been proposed to attempt to overcome this limit, of which the Inverse Scattering Transform (IST)~\cite{ablowitz_ist}, or the Nonlinear Fourier Transform (NFT)~\cite{mansoor_all}, has attracted considerable attention, see~\cite{Turitsyn_overview} for an overview of the advances and perspectives of the NFT for optical communications.

Information transmission using the NFT has been demonstrated both numerically and experimentally in several works, such as~\cite{dong_experiment, le_experiment, le_transoceanic}. For purely discrete spectrum modulation, the spectral efficiencies obtained so far are not very high~\cite{hari_multieig}. In this paper, eigenvalues of higher multiplicity in the discrete spectrum are considered for communication. The theory for these eigenvalues has been developed in~\cite{aktosun_ho,martines_ho}, but its application to communications have to the best of our knowledge not been explored yet. We develop a generalized NFT (GNFT) approach to communications. The GNFT applies to a larger class of signals than the NFT, and thereby provides additional
degrees of freedom that might help to improve communications systems. Our simulations also show that our generalized NFT (GNFT) processing seems to be more robust than NFT for signals with closely-spaced simple eigenvalues, even if they do not perfectly coincide.

The paper is organized as follows. In Section~\ref{sec:model}, we introduce the NLSE model. Section~\ref{sec:nft} briefly describes the NFT. In Section~\ref{sec:ho_math}, we explain the theory of higher multiplicity eigenvalues from~\cite{aktosun_ho,martines_ho}, and we prove some properties of the GNFT. In Section~\ref{sec:algorithms}, we show how to compute the direct and inverse GNFT. Section~\ref{sec:implementation} evaluates the effect of practical impairments for closely-spaced eigenvalues. Section~\ref{sec:experiment} numerically demonstrates information transmission using the GNFT, and Section~\ref{sec:conclusion} concludes the paper.

\textit{Notation:} the subscripts $t$, $z$, and $\lambda$ (and only these) denote partial derivatives with respect to the corresponding variable, e.g. $a_\lambda$ denotes $ {\partial d}/{\partial \lambda}$. Repeated subscripts and parenthesized superscripts denote higher-order derivatives, e.g., $a_{\lambda\lambda}=a^{(2)}={\partial^2 a}/{\partial \lambda^2}$ and $a^{(\ell)}={\partial^\ell a}/{\partial \lambda^\ell}$. In the latter case, the derivative is taken with respect to $\lambda$.

\section{System model}\label{sec:model}
Assuming perfect attenuation compensation, the slowly varying component $Q(Z, T)$ of an electrical field propagating along an optical fiber obeys the NLSE~\cite[Eq. (2.3.46)]{agrawal_nfo}:
\begin{align}
\pderiv{}{Z}Q(Z, T)= & -j\frac{\beta_2}{2}\pderivk{}{T}{2}Q(Z, T) +j\gamma\left|Q(Z, T)\right|^2 Q(Z, T) \nonumber \\ &+N(Z, T)
\label{eq:nlse_analog}
\end{align}
where $Z$ is distance, $T$ is time, $\beta_2$ is the group velocity dispersion (GVD) parameter, and $\gamma$ is the nonlinear coefficient. The distributed noise $N(Z, T)$ satisfies
\begin{equation}
\int_{0}^{Z} N(Z', T)\diff Z'=\sqrt{N_{\mathrm{ASE}}}W(Z, T)
\end{equation}
where $N_{\mathrm{ASE}}$ is the noise spectral density. Note that, unlike~\cite{essiambre_limits}, we do not include the distance in the definition of $N_{\mathrm{ASE}}$. The Wiener process $W(Z, T)$ may be defined as
\begin{equation}
W(Z, T)=\lim\limits_{K\to\infty}\frac{1}{\sqrt{K}}\sum_{k=1}^{\left\lfloor KZ\right\rfloor} W_k(T)
\end{equation}
where the $W_k(T)$ are independent and identically distributed (i.i.d.) circularly symmetric complex Gaussian processes with zero mean, bandwidth $B$, and autocorrelation
\begin{equation}
\mathrm{E}\left[W_k(T)W_k^*(T')\right]=B\sinc\left(B\left(T-T'\right)\right)
\end{equation}
where $\sinc(x)\triangleq\sin\left(\pi x\right)/\left(\pi x\right)$.

\section{The Nonlinear Fourier Transform}\label{sec:nft}
In this section, we briefly introduce the steps involved in the NFT. For more detail, the reader is referred to~\cite{mansoor_all}.

By applying the following change of variables:
\begin{align}
&T=T_0t,\quad Z=2\frac{T_0^2}{\left|\beta_2\right|}z, \quad Q(Z, T)=\frac{1}{T_0}\sqrt{\frac{\left|\beta_2\right|}{\gamma}}q(z, t)
\label{eq:normalize}
\end{align}
the NLSE~\eqref{eq:nlse_analog}, ignoring noise, is normalized to
\begin{equation}
q_z(z, t)=-j\sign\left(\beta_2\right)q_{tt}(z, t) +j2\left|q(z, t)\right|^2 q(z, t) 
\label{eq:nlse}
\end{equation}
and we choose $\beta_2<0$ to focus on the case of anomalous GVD~\cite[p. 131]{agrawal_nfo}. The parameter $T_0$ can be freely chosen. 

The NFT is based on the existence of a Lax pair $(L, M)$ of operators that satisfies
\begin{equation}
{L}_{z}=ML-LM.
\end{equation}
As shown in~\cite[Section 1.4]{ablowitz_ist}, the eigenvalues $\lambda$ of $L$ are invariant in $z$. For the NLSE, the eigenvectors $v$ of $L$ satisfy
\begin{align}
v_z&=Mv \\v_t&=\left({\begin{matrix}
	-j\lambda & q \\ -q^* & j\lambda
	\end{matrix}} \right)v \label{eq:zs}
\end{align}
where
\begin{align}
L(z, t)&=j\left({\begin{matrix}
\frac{\partial}{\partial t} & -q \\ -q^* & -\frac{\partial}{\partial t}
\end{matrix}}\right) \\
M(z, t, \lambda)&=\left({\begin{matrix}
2j\lambda^2-j\left|q\right|^2 & -2\lambda q-jq_t \\ 2\lambda q^*-jq^*_t & -2j\lambda^2+j\left|q\right|^2
\end{matrix}}\right). 
\end{align}
The NFT is calculated by solving the Zakharov-Shabat system~\eqref{eq:zs}. We will often drop the dependence on $z$ to simplify notation. A solution $v^2(t, \lambda)$, bounded in the upper complex half plane ($\lambda\in\mathbb{C}^+$), is obtained using the boundary condition
\begin{equation}
v^2(t, \lambda)\to\left(\begin{matrix}
1\\0
\end{matrix}\right)e^{-j\lambda t}, \ t\to-\infty.
\label{eq:bc}
\end{equation}
The \textit{spectral functions} $a(\lambda)$ and $b(\lambda)$ are obtained as
\begin{subequations}
	\begin{align}
	a(\lambda)&=\lim\limits_{t\to\infty} v_1^2(t, \lambda)e^{j\lambda t} \\
	b(\lambda)&=\lim\limits_{t\to\infty} v_2^2(t, \lambda)e^{-j\lambda t}.
	\end{align}
	\label{eq:ab_lim}
\end{subequations}
The NFT of the signal $q(z, t)$ is made up of two spectra:
\begin{itemize}
	\item the \textit{continuous spectrum} $Q_c(\lambda)=\frac{b(\lambda)}{a(\lambda)}$, for $\lambda\in\mathbb{R}$;
	\item the \textit{discrete spectrum} $Q_d(\lambda_k)=\frac{b(\lambda_k)}{a_\lambda(\lambda_k)}$, for the $K$ eigenvalues $\left\{\lambda_k\in\mathbb{C}^+\colon a(\lambda_k)=0\right\}$.
\end{itemize}

The usefulness of the NFT lies in the fact that, given a signal $q(z, t)$ propagating according to the \textit{noise-free, lossless} NLSE~\eqref{eq:nlse}, its NFT evolves in $z$ according to the following multiplicative relations:
\begin{subequations}
\begin{align}
Q_c(z, \lambda)&=Q_c(0, \lambda)e^{4j\lambda^2 z} \label{eq:Q_cont}\\
\lambda_k(z)&=\lambda_k(0) \\
Q_d(z, \lambda_k)&=Q_d(0, \lambda_k)e^{4j\lambda_k^2 z}. \label{eq:Qd_z}
\end{align}
\end{subequations}

\section{Eigenvalues of higher multiplicity}\label{sec:ho_math}
If $\lambda_k$ is a multiple zero of $a(\lambda)$, then $a_{\lambda}(\lambda_k)=0$, and the above definition of the discrete spectrum is not valid. To the best of our knowledge, all the work on NFT-based communication assumes that all zeros of $a(\lambda)$ are simple, i.e., that all eigenvalues $\lambda_k$ have multiplicity $1$. There has been, however, some work~\cite{aktosun_ho, martines_ho} on the mathematical theory of higher multiplicity eigenvalues.

If the multiplicity of the eigenvalue $\lambda_k$ is $L_k$, we need $L_k$ constants $Q_{k0},\ldots,Q_{k, (L_k-1)}$ to determine the discrete spectrum. In~\cite{martines_ho}, these \textit{norming constants} are defined from the coefficients $r_{k, \ell}$ of the principal part of the Laurent series of $b(\lambda)/a(\lambda)$ around $\lambda_k$:
\begin{equation}
\frac{b(\lambda)}{a(\lambda)}=\frac{r_{k,L_k-1}}{\left(\lambda-\lambda_k\right)^{L_k}}+\cdots +\frac{r_{k,0}}{\left(\lambda-\lambda_k\right)}+\mathcal{O}(1).
\label{eq:Q_kl_series}
\end{equation}
The norming constants $Q_{k,\ell}=j^\ell r_{k,\ell}$ can be calculated as
\begin{equation}
Q_{k,\ell}=\frac{j^\ell}{(L_k-\ell-1)!}\lim\limits_{\lambda\to\lambda_k}\derivk{}{\lambda}{L_k-\ell-1}\left[\left(\lambda-\lambda_k\right)^{L_k}\frac{b(\lambda)}{a(\lambda)}\right].
\label{eq:Q_kl}
\end{equation}

The generalization of the distance evolution equation~\eqref{eq:Qd_z} to the case $L_k\ge 1$ is given by~\cite[Eq. (4.9)]{martines_ho}:
\begin{multline}
\left[\begin{matrix}
Q_{k, (L_k-1)}(z) & \cdots & Q_{k0}(z)
\end{matrix}\right]\\=\left[\begin{matrix}
Q_{k, (L_k-1)}(0) & \cdots & Q_{k0}(0)
\end{matrix}\right]e^{-4j\mathbf{\Lambda}_k^2z}
\label{eq:time_evolution}
\end{multline}
for all $k\inset{1}{K}$, where
\begin{equation}
\mathbf{\Lambda}_k=\left(\begin{matrix}
-j\lambda_k & -1 & 0 & \cdots & 0 \\
0 & -j\lambda_k & -1 & \cdots & 0 \\
\vdots & \vdots & \ddots & \ddots & \vdots  \\
0 & \cdots & 0 & -j\lambda_k & -1 \\
0 & 0 & \cdots & 0 & -j\lambda_k
\end{matrix}\right)\inC{L_k\times L_k}.
\label{eq:Lambda_k}
\end{equation}
We write the GNFT as
\begin{equation}
\mathrm{GNFT}\left\{q(t)\right\}=(Q_c(\lambda), \{\lambda_k\}, \left\{Q_{k\ell}\right\}).
\label{eq:gnft}
\end{equation}

\subsection{Properties of the GNFT}
We prove the following properties in Appendix~\ref{app:properties}.
\begin{enumerate}
	\item \textit{Phase shift}:
	\begin{equation}
	\mathrm{GNFT}\left\{q(t)e^{j\phi_0}\right\}=(Q_c(\lambda)e^{-j\phi_0}, \{\lambda_k\}, \left\{Q_{k\ell}e^{-j\phi_0}\right\}).
	\label{eq:pshift}
	\end{equation}
	\item \textit{Time shift}: if $q'(t)=q(t-t_0)$ then
	\begin{equation}
	\mathrm{GNFT}\left\{q'(t)\right\}=(Q_c'(\lambda), \left\{\lambda_k'\right\}, \left\{Q_{k\ell}'\right\})
	\end{equation}
	satisfies
	\begin{subequations}
	\begin{align}
	& Q_c'(\lambda)=Q_c(\lambda)e^{-2j\lambda t_0} \label{eq:tshift_c} \\
	& \lambda_k'=\lambda_k \label{eq:tshift_lambda} \\
	& \left[\begin{matrix}
	Q_{k, (L_k-1)}' & \cdots & Q_{k0}'
	\end{matrix}\right]\nonumber \\
	&=\left[\begin{matrix}
	Q_{k, (L_k-1)} & \cdots & Q_{k0}
	\end{matrix}\right]e^{2\mathbf{\Lambda}_kt_0}. \label{eq:tshift_d} 
	\end{align}
	\end{subequations}
	\item \textit{Frequency shift}:
	\begin{equation}
	\mathrm{GNFT}\left\{q(t)e^{-2j\omega_0 t}\right\}=(Q_c(\lambda-\omega_0), \{\lambda_k+\omega_0\}, \left\{Q_{k\ell}\right\}).
	\label{eq:fshift}
	\end{equation}
	\item \textit{Time dilation}: for $T>0$
	\begin{equation}
	\mathrm{GNFT}\left\{\frac{1}{T}q\left(\frac{t}{T}\right)\right\}=\left(Q_c(T\lambda), \left\{\frac{\lambda_k}{T}\right\}, \left\{\frac{Q_{k\ell}}{T^{\ell+1}}\right\}\right).
	\label{eq:tdil}
	\end{equation}
	\item \textit{Parseval's theorem}:
	\begin{align}
	\int_{-\infty}^{\infty} \left|q(t)\right|^2\diff t=&\frac{1}{\pi}\int_{-\infty}^{\infty}\log\left(1+\left|Q_c(\lambda)\right|^2\right)\diff{\lambda} \nonumber \\ & +4\sum_{k=0}^{K}L_k\Imag{\lambda_k}.
	\label{eq:parseval}
	\end{align}
\end{enumerate}

\section{Numerical computation of the (I)GNFT}~\label{sec:algorithms}
We extend existing numerical algorithms that compute the (I)NFT to include multiple eigenvalues.
\subsection{Direct GNFT}
Most algorithms that compute the direct NFT discretize the Zakharot-Shabat system~\eqref{eq:zs} to find $a(\lambda)$ and $b(\lambda)$ from~\eqref{eq:ab_lim}. Let $u=(u_1, u_2)^T$, where $u_1(t, \lambda)=v_1^2(t, \lambda)e^{j\lambda t}$ and $u_2(t, \lambda)=v_2^2(t, \lambda)e^{-j\lambda t}$. Then from~\eqref{eq:zs} and~\eqref{eq:ab_lim} we have
\begin{align}
u_t(t, \lambda)&={\left(\begin{matrix}
0 & q(t)e^{2j\lambda t} \\ -q^*(t)e^{-2j\lambda t} & 0
\end{matrix}\right)}u(t, \lambda)
\label{eq:zs_ab} \\
\left(\begin{matrix}a(\lambda) \\ b(\lambda)\end{matrix}\right)&=\lim\limits_{t\to\infty}\left(\begin{matrix}u_1(t, \lambda) \\ u_2(t, \lambda) \end{matrix}\right).
\end{align}
To compute the GNFT of $q(t)$, we discretize the time axis for $t\in\left[t_1,t_2\right]$. Let $t_n=t_1+n\epsilon$, $q_n=q(t_n)$, where $n\inset{0}{N-1}$, $N$ is the number of samples, and $\epsilon=(t_2-t_1)/(N-1)$ is the step size. Similarly, let $u[n]=u(t_1+n\epsilon, \lambda)$. Starting at $u[0]=(1, 0)^T$ (see~\eqref{eq:bc}), the following update step is applied iteratively:
\begin{equation}
u[n+1]=A[n]u[n],\ \  n\inset{0}{N-2}
\label{eq:update}
\end{equation}
and we have $a(\lambda)=u_1[N-1]$ and $b(\lambda)=u_2[N-1]$. The kernel $A[n]$ depends on the discretization algorithm~\cite{mansoor_all}. We consider the trapezoidal kernel proposed in~\cite{aref_control_detection}:
\begin{equation}
A[n]=\left(\begin{matrix}
\cos\left(\left|q_n\right|\epsilon\right) & \sin\left(\left|q_n\right|\epsilon\right)e^{j\left(\theta_n+2\lambda t_n\right)} \\ -\sin\left(\left|q_n\right|\epsilon\right)e^{-j\left(\theta_n+2\lambda t_n\right)} & \cos\left(\left|q_n\right|\epsilon\right)
\end{matrix}\right)
\label{eq:trap}
\end{equation}
where $\theta_n=\arg q_n$. However, the following analysis is valid for any kernel $A[n]$. 
To obtain the norming constants $Q_{k\ell}$, we need to calculate higher order $\lambda$-derivatives of $a(\lambda)$ and $b(\lambda)$. We obtain bounds on the order of the required derivatives.
\begin{lemma}\label{th:a_lambda}
	The value of $q_{k\ell}$ in~\eqref{eq:Q_kl} depends on $\lambda_k$ only through $a^{(m)}(\lambda_k)$ for $m\inset{L_k}{2L_k-\ell-1}$ and $b^{(n)}(\lambda_k)$ for $n\inset{0}{L_k-\ell-1}$.
\end{lemma}
\begin{proof}
	See Appendix~\ref{app:a_lambda}.
\end{proof}
For an eigenvalue of multiplicity $L_k$, we compute the first $2L_k-1$ derivatives of $u[N-1]$ by setting the following additional initial conditions and update steps
\begin{subequations}
\begin{align}
u^{(m)}[0]&=\left(\begin{matrix}
0\\0
\end{matrix}\right), \qquad m\inset{1}{2L_k-1} \\
u^{(m)}[n+1]&=\sum_{r=0}^{m}\binom{m}{r}A^{(r)}[n]u^{(m-r)}[n]
\end{align}%
\label{eq:update_deriv}%
\end{subequations}%
where $A^{(r)}[n]$, the $r$-th order $\lambda$-derivative of $A[n]$, is obtained in closed form. 
Once we have the required values of $a, b$ and their derivatives, we use~\eqref{eq:Q_kl} to compute the norming constants. In~\eqref{eq:Q_kl}, the derivative is evaluated in closed form, and then L'H\^opital's rule is applied repetitively to obtain an expression for $t_{k\ell}$ that depends only on nonzero derivatives of $a$. See~\eqref{eq:g}-\eqref{eq:rr} for details. For $L_k=2$, this gives
\begin{subequations}	
\begin{align}
Q_{k1}&=\frac{j2b(\lambda_k)}{a_{\lambda\lambda}(\lambda_k)} \\
Q_{k0}&=\frac{2b_{\lambda}(\lambda_k)}{a_{\lambda\lambda}(\lambda_k)}-\frac{2}{3}\frac{b(\lambda_k)a_{\lambda\lambda\lambda}(\lambda_k)}{a_{\lambda\lambda}(\lambda_k)^2}.
\end{align}
\label{eq:norming_constants_L2}
\end{subequations}

\subsubsection*{Forward-Backward Method}
This technique was proposed in~\cite{aref_control_detection} to improve numerical stability. We write~\eqref{eq:update} as
\begin{equation}
\left(\begin{matrix}
a(\lambda) \\ b(\lambda)
\end{matrix}\right)=A[N-1]\cdots A[1]A[0]\left(\begin{matrix}
1 \\ 0
\end{matrix}\right)=RL\left(\begin{matrix}
1 \\ 0
\end{matrix}\right)
\label{eq:RL}
\end{equation}
where $R=A[N-1]\cdots A[n_0]$ and $L=A[n_0-1]\cdots A[0]$, and $n_0$ is chosen according to some criterion to minimize the numerical error. The iterative procedure~\eqref{eq:update} is run \textit{forward} up to $n_0-1$ to obtain
\begin{equation}
\left(\begin{matrix}
l_1 \\ l_2
\end{matrix}\right)=L\left(\begin{matrix}
1 \\ 0
\end{matrix}\right)=\left(\begin{matrix}
L_{11} \\ L_{21}
\end{matrix}\right)
\end{equation}
and \textit{backward} from $r[N-1]=(0, 1)^T$ down to $r[n_0-1]$:
\begin{equation}
\left(\begin{matrix}
r_1 \\ r_2
\end{matrix}\right)=R^{-1}\left(\begin{matrix}
0 \\ 1
\end{matrix}\right)=\left(\begin{matrix}
-R_{12} \\ R_{11}
\end{matrix}\right).
\label{eq:r}
\end{equation}
The kernel $A[n]^{-1}$ is used to compute~\eqref{eq:r}: for the trapezoidal case this amounts to replacing $\epsilon$ with $-\epsilon$ in~\eqref{eq:trap}. Note that~\eqref{eq:r} is valid only for kernels with unit determinant.

Using~\eqref{eq:update_deriv}, we obtain $r_1$, $r_2$, $l_1$, $l_2$, and their derivatives up to order $2L_k-1$. From~\eqref{eq:RL} we have
\begin{equation}
a(\lambda)=R_{11}L_{11}+R_{12}L_{21}
\end{equation}
and we compute
\begin{equation}
a^{(\ell)}(\lambda_k)=\sum_{m=0}^{\ell}\binom{\ell}{m}\left(r_2^{(m)}l_1^{(\ell-m)}-r_1^{(m)}l_2^{(\ell-m)}\right).
\label{eq:a_fb}
\end{equation}
To obtain $b^{(\ell)}(\lambda_k)$, note that
\begin{align}
b(\lambda_k)&=R_{21}L_{11}+R_{22}L_{21} \nonumber \\
&=\frac{R_{21}}{R_{11}}\left(R_{11}L_{11}+R_{12}L_{21}\right)+\frac{L_{21}}{R_{11}} \nonumber \\
&=\frac{R_{21}}{R_{11}}a(\lambda_k)+\frac{L_{21}}{R_{11}}
\label{eq:b}
\end{align}
where we used $R_{22}=\left(1+R_{12}R_{21}\right)/R_{11}$. The $\ell$-th derivative of the left summand in~\eqref{eq:b} is $0$ for $\ell\le L_k-1$, because $a^{(\ell)}(\lambda_k)=0$ for $\ell\le L_k-1$ . Therefore, we have
\begin{equation}
b^{(\ell)}(\lambda_k)=\left.\derivk{}{\lambda}{\ell}\frac{l_2}{r_2}\right|_{\lambda=\lambda_k}
\label{eq:b_fb}
\end{equation}
which can be written in closed form using~\eqref{eq:quotient} below. Equations~\eqref{eq:b_fb} and~\eqref{eq:a_fb}, together with~\eqref{eq:Q_kl} or~\eqref{eq:norming_constants_L2}, let us compute the GNFT from the forward-backward method.

\subsection{Inverse GNFT}
The inverse GNFT can be computed by solving the generalized Gelfand-Levitan-Marchenko equation (GLME)~\cite{le_glme}:
\begin{multline}
K(t, y)-\mathrm{\Omega}^*(t+y) \\ +\int_{t}^{\infty}\diff x\int_{t}^{\infty}\diff s\;K(t, s)\mathrm{\Omega}(s+x)^*\mathrm{\Omega}(x+y)=0.
\label{eq:glme}
\end{multline}
The kernel $\mathrm{\Omega}(y)$ is given by~\cite{aktosun_ho}:
\begin{equation}
\mathrm{\Omega}(y)=\frac{1}{2\pi}\int_{-\infty}^{\infty}Q_c(\lambda)e^{j\lambda y}\diff\lambda+\sum_{k=1}^{K}\sum_{\ell=0}^{L_k-1}Q_{k\ell}\frac{y^\ell}{\ell!}e^{j\lambda_k y}.
\label{eq:Omega}
\end{equation}
The inverse GNFT is then obtained as
\begin{equation}
q(t)=-2K(t, t).
\label{eq:q}
\end{equation}
The derivation of~\eqref{eq:glme}-\eqref{eq:q} is given in~\cite{ablowitz_ist} and is based on expressing $v^2(t)=v^2(-\infty)+\int_{-\infty}^{t} K(t, s)e^{-j\lambda s}\diff{s}$ and substituting in~\eqref{eq:zs}. A numerical procedure to solve~\eqref{eq:glme} is given in~\cite[Section 4.2]{le_glme}: it suffices to replace $F(y)$ by $\mathrm{\Omega}(y)$.

When there is no continuous spectrum, a closed-form expression is given in~\cite{aktosun_ho} for the generalized $K$-solitons:
\begin{align}
&q(z, t)= \nonumber\\ &-2\mathbf{b}^He^{-\mathbf{\Lambda}^Ht}\left(\mathbf{I}+\mathbf{M}(z, t)\mathbf{N}(t)\right)^{-1}e^{-\mathbf{\Lambda}^Ht+4j\left(\mathbf{\Lambda}^H\right)^2z}\mathbf{c}
\label{eq:q_closed_form}
\end{align}
where
\begin{equation}
\mathbf{\Lambda}=\left(\begin{matrix}
\mathbf{\Lambda}_1 & 0 & \cdots & 0 \\
0 & \mathbf{\Lambda}_2 & \cdots & 0 \\
\vdots & \vdots & \ddots & \vdots \\
0 & \cdots & 0 & \mathbf{\Lambda}_K
\end{matrix}\right)
\end{equation}
\begin{equation*}
{\setlength{\arraycolsep}{3pt}
\mathbf{b}=\left(\begin{matrix}
\mathbf{b}_1^T & \cdots & \mathbf{b}_K^T
\end{matrix}\right)^T, \  \mathbf{b}_k=\left(\begin{matrix}
0 & \cdots & 0 & 1
\end{matrix}\right)^T\in\left\{0,1\right\}^{L_k\times 1}}
\end{equation*}
\begin{equation*}
\mathbf{c}=\left(\begin{matrix}
\mathbf{c}_1^T & \cdots & \mathbf{c}_K^T
\end{matrix}\right), \quad \mathbf{c}_k=\left(\begin{matrix}
Q_{k, (L_k-1)}^* & \cdots & Q_{k0}^*
\end{matrix}\right)^T
\end{equation*}
\begin{align}
&\mathbf{M}(z, t)=\int_{t}^{\infty}e^{-\mathbf{\Lambda}^H s+4j\left(\mathbf{\Lambda}^H\right)^2 z}\mathbf{c}\mathbf{c}^He^{-\mathbf{\Lambda} s-4j\mathbf{\Lambda}^2 z}\diff s
\label{eq:M} \\
&\mathbf{N}(t)=\int_t^\infty e^{-\mathbf{\Lambda}x}\mathbf{b}\mathbf{b}^He^{-\mathbf{\Lambda}^H x}\diff x.
\label{eq:N}
\end{align}
$\mathbf{\Lambda}_k$ is given by~\eqref{eq:Lambda_k}, and $\mathbf{I}$ is an identity matrix of size $\sum_k L_k$. The integrals~\eqref{eq:M} and~\eqref{eq:N} must be computed numerically.

\subsection{Example: Double Soliton (DS)}
From~\eqref{eq:q_closed_form}, a soliton with a second order eigenvalue at $\lambda=\xi+j\eta$ and norming constants $Q_{11}$ and $Q_{10}$ is given by
\begin{equation}
q(z, t)=\frac{h(z, t)}{f(z, t)}
\label{eq:double_soliton}
\end{equation}
where
\begin{multline}
h(z, t)=-j4\eta e^{-j\arg Q_{11}}e^{-j2\xi t}e^{-j4\left(\xi^2-\eta^2\right)z}\left\{\vphantom{\left[\left|Q_{11}\right|^2\right]} \right.  \\ \left. e^{-X}\left[-\left|Q_{11}\right|^2\left(2\eta t+8\eta\left(\xi+j\eta\right) z+2\right)-\eta Q_{11}^*Q_{10}\right] \right. \\ \left.+e^{X}\left[\left|Q_{11}\right|^2\left(2\eta t+8\eta\left(\xi-j\eta\right) z\right)+\eta Q_{11}Q_{10}^*\right]\right\}
\end{multline}
\begin{align}
f(z, t)=&\left|Q_{11}\right|^2\left[\cosh\left(2X\right)+1\right] \nonumber \\ & +2\left|Q_{10}\eta+Q_{11}\left(2\eta t+8\eta\left(\xi+j\eta\right) z+1\right)\right|^2
\end{align}
\begin{equation}
X=2\eta t+8\eta\xi z-\log\frac{|Q_{11}|}{4\eta^2}.
\end{equation}
We refer to this soliton as a \textit{double soliton} (DS). The evolution of the norming constants~\eqref{eq:time_evolution} reduces to
\begin{subequations}
\begin{align}
Q_{11}(z)&=Q_{11}(0)e^{4j\lambda^2 z} \\
Q_{10}(z)&=\left(Q_{10}(0)+8\lambda zQ_{11}(0)\right)e^{4j\lambda^2 z}.
\label{eq:Q_10_ev}
\end{align}
\end{subequations}
Note from~\eqref{eq:double_soliton} that a DS does \textit{not} exhibit periodic (breathing) behavior in $z$. The monotonic growth of some norming constants with $z$ suggests that, generally, solitons with higher multiplicity eigenvalues do not breathe. 

The following expression for the \textit{center time} of the DS was obtained empirically and seems to be valid based on our simulations, but we have not found a proof:
\begin{equation}
\frac{\int_{-\infty}^{\infty}t\left|q(z, t)\right|^2}{\int_{-\infty}^{\infty}\left|q(z, t)\right|^2}=\frac{1}{2\eta}\log\left(\frac{\left|Q_{11}(z)\right|}{4\eta^2}\right).
\label{eq:center_time}
\end{equation}
This definition of the center time has proven useful for the analysis of the propagation of trains of ordinary solitons~\cite{prilepsky_train}.
%

\section{Effect of implementation limitations}\label{sec:implementation}
Eigenvalues of higher multiplicity require the pulse to have the exact shape given by, e.g.,~\eqref{eq:q_closed_form} or ~\eqref{eq:double_soliton}. Any deviation caused by practical limitations such as discretization, truncation, attenuation or noise splits an eigenvalue $\lambda_k$ of multiplicity $L_k$ into $L_k$ very closely spaced eigenvalues $\lambda_{k,\ell}$. However, the computation of the spectral amplitudes
\begin{equation}
Q_d(\lambda_{k,\ell})=\frac{b(\lambda_{k,\ell})}{a_{\lambda}(\lambda_{k,\ell})}
\end{equation}
is unstable, because the denominator is close to $0$. The norming constants $Q_{k,\ell}$ of the GNFT seem to be more stable, even when the actual signal has two closely spaced eigenvalues instead of one eigenvalue of multiplicity $2$. Furthermore, the search algorithm might not be able to distinguish eigenvalues that are too close. We next illustrate these observations.

\subsection{Effect of pulse truncation}\label{sec:truncation}
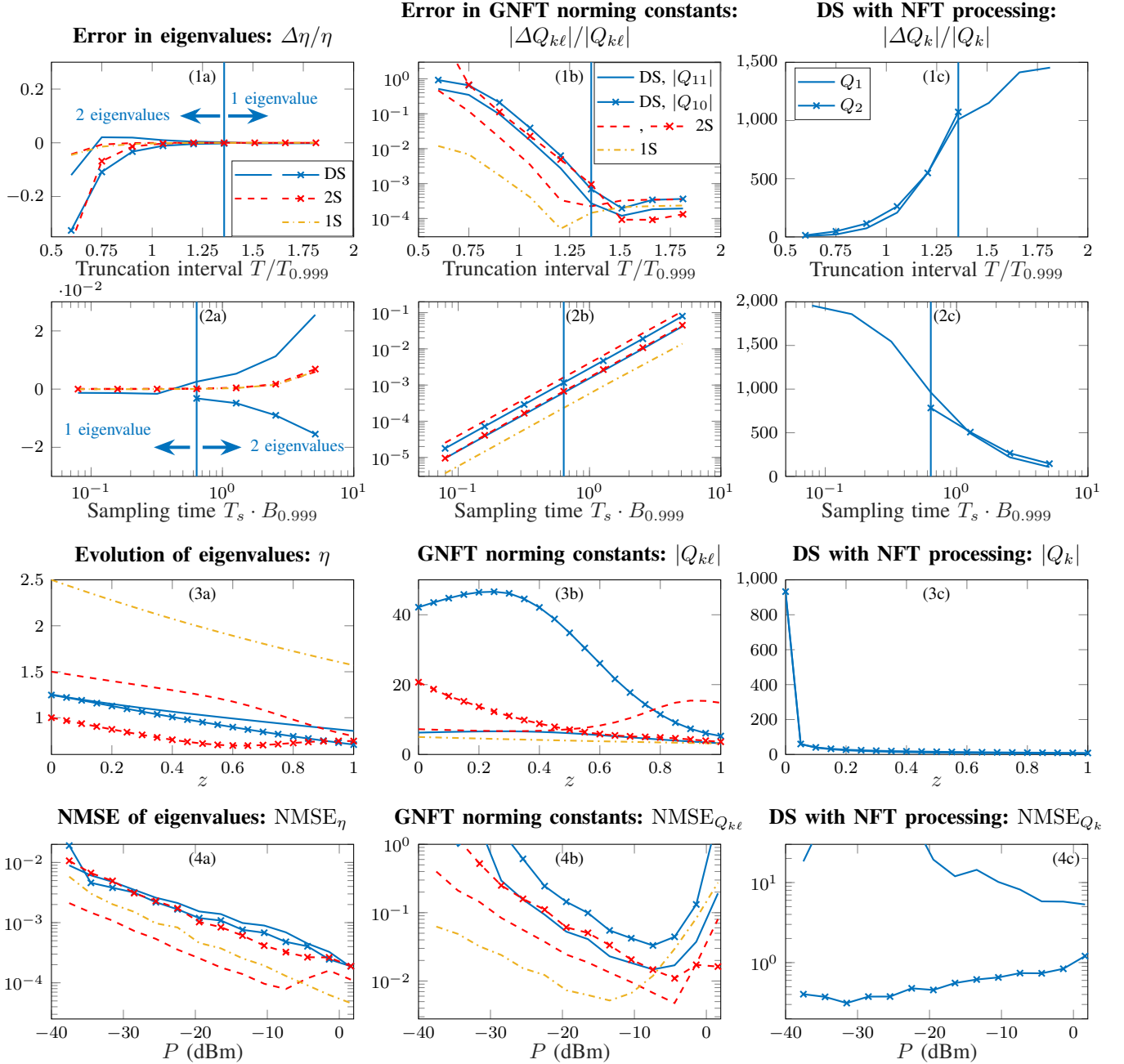
\begin{figure*}[!t]\centering
	\begin{tikzpicture}[node distance=3em]
		\newlength{\verticaldistance}
		\setlength{\verticaldistance}{3em}
		\setlength{\figurewidth}{0.285\textwidth}
		\setlength{\figureheight}{0.55\figurewidth}
		\pgfplotsset{yticklabel style={text width=2em,align=right}}
%
%
\definecolor{mycolor1}{rgb}{0.00000,0.44700,0.74100}%
\definecolor{mycolor2}{rgb}{1,0,0}%
\definecolor{mycolor3}{rgb}{0.92900,0.69400,0.12500}%
\definecolor{mycolor4}{rgb}{0.49400,0.18400,0.55600}%
\definecolor{mycolor5}{rgb}{0.46600,0.67400,0.18800}%
%

\begin{axis}[%
name=eig_trunc_imag,
width=0.951\figurewidth,
height=\figureheight,
scale only axis,
unbounded coords=jump,
xmin=0.5,
xmax=2,
xtick={0.5, 0.75, 1, 1.25, 1.5, 1.75, 2},
xlabel style={yshift=0.5em, font=\color{white!15!black}},
xlabel={Truncation interval $T/T_{0.999}$},
ymin=-0.35,
ymax=0.3,
axis background/.style={fill=white},
 title style={yshift=-0.3em, font=\bfseries},
 title={Error in eigenvalues: $\Delta\eta/\eta$},
legend columns=2,
legend style={at={(0.999, 0.001)}, anchor=south east, legend cell align=left, align=left, draw=white!15!black, inner sep=0.1em},
]
\addplot [color=mycolor1, thick]
table[row sep=crcr]{%
	0.598170731707317	-0.120227259168262\\
	0.74997921286031	0.0206497658894589\\
	0.901787694013304	0.0194161504674392\\
	1.0535961751663	0.00986263767645816\\
	1.20540465631929	0.00423791588821878\\
	1.35721313747228	0.0021169206688592\\
	1.50902161862528	-0.00218860477159666\\
	1.66083009977827	-0.00176963197995743\\
	1.81263858093126	-0.00171824029826322\\
};
\addlegendentry{}

\addplot [color=mycolor1, thick, mark=x]
table[row sep=crcr]{%
0.598170731707317	-0.326167545368237\\
0.74997921286031	-0.108463582241827\\
0.901787694013304	-0.0327444599162444\\
1.0535961751663	-0.0112034689169022\\
1.20540465631929	-0.0043841100807926\\
1.35721313747228	-0.00237202126476745\\
1.50902161862528	nan\\
1.66083009977827	nan\\
1.81263858093126	nan\\
};
\addlegendentry{DS}

\addplot [color=mycolor2, dashed, thick]
table[row sep=crcr]{%
	0.598170731707317	-0.0428451428782443\\
	0.74997921286031	-0.0064164183926804\\
	0.901787694013304	-0.000881072646176614\\
	1.0535961751663	0\\
	1.20540465631929	0\\
	1.35721313747228	0\\
	1.50902161862528	0\\
	1.66083009977827	0\\
	1.81263858093126	0\\
};
\addlegendentry{}

\addplot [color=mycolor2, dashed, thick, mark=x, mark options={solid}]
table[row sep=crcr]{%
	0.598170731707317	-0.371647746890672\\
	0.74997921286031	-0.0683811104577651\\
	0.901787694013304	-0.0138626203065837\\
	1.0535961751663	-0.00271792964697093\\
	1.20540465631929	0\\
	1.35721313747228	0\\
	1.50902161862528	0\\
	1.66083009977827	0\\
	1.81263858093126	0\\
};
\addlegendentry{2S}
\addlegendimage{empty legend}
\addlegendentry{}
\addplot [color=mycolor3, dashdotted, thick]
table[row sep=crcr]{%
	0.598170731707317	-0.0455358278062457\\
	0.74997921286031	-0.0142980648073264\\
	0.901787694013304	-0.0044545255501049\\
	1.0535961751663	-0.00131592407894239\\
	1.20540465631929	0\\
	1.35721313747228	0\\
	1.50902161862528	0\\
	1.66083009977827	0\\
	1.81263858093126	0\\
};
\addlegendentry{1S}

\draw [color=mycolor1, thick] (axis cs:1.35721313747228, -0.35) -- (1.35721313747228, 0.3);
\draw [color=mycolor1, thick, line width=0.15em, arrows={-Stealth}] (1.33721313747228, 0.1) -- (1.13721313747228, 0.1) node [left] {\textcolor{mycolor1}{2 eigenvalues}};
\draw [color=mycolor1, thick, line width=0.15em, arrows={-Stealth}] (1.37721313747228, 0.1) -- (1.57721313747228, 0.1) node [anchor=250] {\textcolor{mycolor1}{1 eigenvalue}};
\node at (1.25, 0.3) [anchor=north] {(1a)};
\end{axis}
%
%
\definecolor{mycolor1}{rgb}{0.00000,0.44700,0.74100}%
\definecolor{mycolor2}{rgb}{1,0,0}%
\definecolor{mycolor3}{rgb}{0.92900,0.69400,0.12500}%
\definecolor{mycolor4}{rgb}{0.49400,0.18400,0.55600}%
\definecolor{mycolor5}{rgb}{0.46600,0.67400,0.18800}%
%

\begin{axis}[%
name=eig_trunc_normconst,
at=(eig_trunc_imag.east),
anchor=west,
xshift=\nodedistance,
width=0.951\figurewidth,
height=\figureheight,
scale only axis,
xmin=0.5,
xmax=2,
xtick={0.5, 0.75, 1, 1.25, 1.5, 1.75, 2},
xlabel style={yshift=0.5em, font=\color{white!15!black}},
xlabel={Truncation interval $T/T_{0.999}$},
ymode=log,
yminorticks=true,
ymin=3e-5,
ymax=3,
max space between ticks=20,
axis background/.style={fill=white},
title style={yshift=-0.3em, font=\bfseries, align=center},
title={Error in GNFT norming constants:\\ $|\Delta Q_{k\ell}|/|Q_{k\ell}|$},
legend style={at={(1, 1)}, anchor=north east, legend cell align=left, align=left, draw=white!15!black, inner sep=0.1em}
]
\addplot [color=mycolor1, thick]
table[row sep=crcr]{%
0.598170731707317	0.515680761558703\\
0.74997921286031	0.346696443452872\\
0.901787694013304	0.096707642469018\\
1.0535961751663	0.0173207554483102\\
1.20540465631929	0.00274661118140301\\
1.35721313747228	0.000272413728465466\\
1.50902161862528	0.00012003501078155\\
1.66083009977827	0.000183782546801297\\
1.81263858093126	0.000193723092290128\\
};
\addlegendentry{DS, $|Q_{11}|$}

\addplot [color=mycolor1, thick, mark=x]
table[row sep=crcr]{%
0.598170731707317	0.923340319947913\\
0.74997921286031	0.651282232842732\\
0.901787694013304	0.21030921837024\\
1.0535961751663	0.0392136057661241\\
1.20540465631929	0.00630128970246154\\
1.35721313747228	0.000693388159986041\\
1.50902161862528	0.000196790706926557\\
1.66083009977827	0.000341912876976515\\
1.81263858093126	0.000364482947193029\\
};
\addlegendentry{DS, $|Q_{10}|$}

\addplot [color=mycolor2, thick, dashed]
table[row sep=crcr]{%
0.598170731707317	0.46612430490443\\
0.74997921286031	0.118012327985121\\
0.901787694013304	0.0209544310265662\\
1.0535961751663	0.00346912934176348\\
1.20540465631929	0.000334659006052344\\
1.35721313747228	0.000223778851981839\\
1.50902161862528	0.000328619200990366\\
1.66083009977827	0.000347476671871575\\
1.81263858093126	0.000353135744114454\\
};

\addplot [color=mycolor2, thick, dashed, mark=x, mark options={solid}, forget plot]
table[row sep=crcr]{%
0.598170731707317	25.4059975460101\\
0.74997921286031	0.677035228180539\\
0.901787694013304	0.112567900562754\\
1.0535961751663	0.0228949633190717\\
1.20540465631929	0.00490286758976334\\
1.35721313747228	0.000947605040145253\\
1.50902161862528	9.28393597058215e-05\\
1.66083009977827	9.20060973064646e-05\\
1.81263858093126	0.000132457543245711\\
};\label{2S_2}
\addlegendentry{, \begin{tikzpicture}\draw[dashed, thick, color=mycolor2] (0em,0em) -- (2.2em,0em); \draw[thick, color=mycolor2] (0.9em,0.2em) -- (1.3em,-0.2em);
	\draw[thick, color=mycolor2] (0.9em,-0.2em) -- (1.3em,0.2em);
 \end{tikzpicture} 2S}

\addplot [color=mycolor3, thick, dashdotted]
table[row sep=crcr]{%
	0.598170731707317	0.0119178611361189\\
	0.74997921286031	0.00684388360244022\\
	0.901787694013304	0.00171824948370016\\
	1.0535961751663	0.000407859774424146\\
	1.20540465631929	5.10108496309059e-05\\
	1.35721313747228	0.00014527671100808\\
	1.50902161862528	0.000211771951554773\\
	1.66083009977827	0.00022641104114669\\
	1.81263858093126	0.000235319073476425\\
};
\addlegendentry{1S}

\draw [color=mycolor1, thick] (1.35721313747228, 3e-5) -- (1.35721313747228, 3);
\node at (1.25, 3) [anchor=north] {(1b)};
\end{axis}
%
%
\definecolor{mycolor1}{rgb}{0.00000,0.44700,0.74100}%
\definecolor{mycolor2}{rgb}{1,0,0}%
\definecolor{mycolor3}{rgb}{0.92900,0.69400,0.12500}%
\definecolor{mycolor4}{rgb}{0.49400,0.18400,0.55600}%
\definecolor{mycolor5}{rgb}{0.46600,0.67400,0.18800}%
%

\begin{axis}[%
name=eig_trunc_specamp,
at=(eig_trunc_normconst.east),
anchor=west,
xshift=\nodedistance,
width=0.951\figurewidth,
height=\figureheight,
scale only axis,
unbounded coords=jump,
xmin=0.5,
xmax=2,
xtick={0.5, 0.75, 1, 1.25, 1.5, 1.75, 2},
xlabel style={yshift=0.5em, font=\color{white!15!black}},
xlabel={Truncation interval $T/T_{0.999}$},
ymin=0,
ymax=1500,
axis background/.style={fill=white},
title style={yshift=-0.3em, font=\bfseries, align=center},
title={DS with NFT processing:\\ $|\Delta Q_k|/|Q_k|$},
legend style={at={(0.03, 0.97)}, anchor=north west, legend cell align=left, align=left, draw=white!15!black, inner sep=0.1em}
]
\addplot [color=mycolor1, thick]
table[row sep=crcr]{%
	0.598170731707317	7.9717608505307\\
	0.74997921286031	21.0582373234987\\
	0.901787694013304	74.1701310957171\\
	1.0535961751663	208.878968818956\\
	1.20540465631929	550.987170501297\\
	1.35721313747228	1009.54115683297\\
	1.50902161862528	1151.47014239183\\
	1.66083009977827	1413.61007842376\\
	1.81263858093126	1454.31159864653\\
};
\addlegendentry{$Q_1$}

\addplot [color=mycolor1, thick, mark=x]
table[row sep=crcr]{%
0.598170731707317	14.0697264494945\\
0.74997921286031	47.8339277675339\\
0.901787694013304	115.08714800546\\
1.0535961751663	262.0107523445\\
1.20540465631929	549.723455186477\\
1.35721313747228	1073.7045371667\\
1.50902161862528	nan\\
1.66083009977827	nan\\
1.81263858093126	nan\\
};
\addlegendentry{$Q_2$}

\draw [color=mycolor1, thick] (1.35721313747228, 0) -- (1.35721313747228, 1500);

%
%

\node at (1.25, 1500) [anchor=north] {(1c)};
\end{axis}
%
%
\definecolor{mycolor1}{rgb}{0.00000,0.44700,0.74100}%
\definecolor{mycolor2}{rgb}{1,0,0}%
\definecolor{mycolor3}{rgb}{0.92900,0.69400,0.12500}%
\definecolor{mycolor4}{rgb}{0.49400,0.18400,0.55600}%
\definecolor{mycolor5}{rgb}{0.46600,0.67400,0.18800}%
%

\begin{axis}[%
name=eig_discr_imag,
at=(eig_trunc_imag.south),
anchor=north,
yshift=-\verticaldistance,
width=0.951\figurewidth,
height=\figureheight,
scale only axis,
unbounded coords=jump,
xmode=log,
xmin=0.05,
xmax=10,
xminorticks=true,
xlabel style={yshift=0.5em, font=\color{white!15!black}},
xlabel={Sampling time $T_s\cdot B_{0.999}$},
ymin=-0.03,
ymax=0.03,
axis background/.style={fill=white},
legend style={at={(0.03, 0.97)}, anchor=north west, legend cell align=left, align=left, draw=white!15!black}
]
\addplot [color=mycolor1, thick]
table[row sep=crcr]{%
	0.0795378169852995	-0.00128985026173307\\
	0.159075633970599	-0.00135476528356051\\
	0.318151267941198	-0.00161965745067736\\
	0.636302535882396	0.00255562561108462\\
	1.27260507176479	0.00528776901095434\\
	2.54521014352958	0.0112936084925536\\
	5.09042028705917	0.0255333846545279\\
};
\addlegendentry{Double soliton}

\addplot [color=mycolor1, thick, forget plot, mark=x]
table[row sep=crcr]{%
	0.0795378169852995	nan\\
	0.159075633970599	nan\\
	0.318151267941198	nan\\
	0.636302535882396	-0.00319579131763632\\
	1.27260507176479	-0.00482418763563146\\
	2.54521014352958	-0.00901006296333193\\
	5.09042028705917	-0.0154645076646483\\
};

\addplot [color=mycolor2, thick, dashed]
table[row sep=crcr]{%
	0.0795378169852995	0.000102219942267112\\
	0.159075633970599	0.000106406504189138\\
	0.318151267941198	0.000123165778396513\\
	0.636302535882396	0.000190292512475043\\
	1.27260507176479	0.000459490904718122\\
	2.54521014352958	0.00154810591125676\\
	5.09042028705917	0.00591381150599653\\
};
\addlegendentry{2-soliton}

\addplot [color=mycolor2, thick, dashed, forget plot, mark=x, mark options={solid}]
table[row sep=crcr]{%
	0.0795378169852995	1.55816204605586e-06\\
	0.159075633970599	5.78263744599461e-06\\
	0.318151267941198	2.6270462046174e-05\\
	0.636302535882396	0.00010500406860503\\
	1.27260507176479	0.000421481141747559\\
	2.54521014352958	0.00170594750629816\\
	5.09042028705917	0.00688332291221738\\
};

\addplot [color=mycolor3, dashdotted, thick]
table[row sep=crcr]{%
0.0795378169852995	-2.90585321922165e-05\\
0.159075633970599	-2.46813426192816e-05\\
0.318151267941198	-7.16549883339468e-06\\
0.636302535882396	6.30049871954697e-05\\
1.27260507176479	0.000343761450203495\\
2.54521014352958	0.00147624723199122\\
5.09042028705917	0.0060731494114755\\
};
\addlegendentry{1-soliton}
\legend{}
\draw [thick, color=mycolor1] (0.636302535882396, -0.03) -- (0.636302535882396, 0.03);
\draw [color=mycolor1, thick, line width=0.15em, arrows={-Stealth}] (0.566302535882396, -0.02) -- (0.3, -0.02) node [anchor=340] {\textcolor{mycolor1}{1 eigenvalue}};
\draw [color=mycolor1, thick, line width=0.15em, arrows={-Stealth}] (0.706302535882396, -0.02) -- (1.4, -0.02) node [right] {\textcolor{mycolor1}{2 eigenvalues}};
\node at (0.85, 0.03) [anchor=north] {(2a)};
\end{axis}
%
%
\definecolor{mycolor1}{rgb}{0.00000,0.44700,0.74100}%
\definecolor{mycolor2}{rgb}{1,0,0}%
\definecolor{mycolor3}{rgb}{0.92900,0.69400,0.12500}%
\definecolor{mycolor4}{rgb}{0.49400,0.18400,0.55600}%
\definecolor{mycolor5}{rgb}{0.46600,0.67400,0.18800}%
%

\begin{axis}[%
name=eig_discr_normconst,
at=(eig_trunc_normconst.south),
anchor=north,
yshift=-\verticaldistance,
width=0.951\figurewidth,
height=\figureheight,
scale only axis,
xmode=log,
xmin=0.05,
xmax=10,
xminorticks=true,
xlabel style={yshift=0.5em, font=\color{white!15!black}},
xlabel={Sampling time $T_s\cdot B_{0.999}$},
ymode=log,
yminorticks=true,
ymin=3e-6,
ymax=0.2,
max space between ticks=20,
axis background/.style={fill=white},
legend style={at={(0.97, 0.5)}, anchor=east, legend cell align=left, align=left, draw=white!15!black}
]
\addplot [color=mycolor1, thick]
table[row sep=crcr]{%
	0.0795378169852995	9.46798944497118e-06\\
	0.159075633970599	3.83383777413826e-05\\
	0.318151267941198	0.000153894661662832\\
	0.636302535882396	0.000616847403333765\\
	1.27260507176479	0.00247658053336551\\
	2.54521014352958	0.010013615082852\\
	5.09042028705917	0.0415596215196061\\
};
\addlegendentry{DS, $|Q_{11}|$}

\addplot [color=mycolor1, thick, mark=x]
table[row sep=crcr]{%
0.0795378169852995	1.77854839378609e-05\\
0.159075633970599	7.22336994217809e-05\\
0.318151267941198	0.000290204091752976\\
0.636302535882396	0.00116387014272558\\
1.27260507176479	0.00467871340051247\\
2.54521014352958	0.0189974483182297\\
5.09042028705917	0.0801061656627365\\
};
\addlegendentry{DS, $|Q_{11}|$}

\addplot [color=mycolor2, thick, dashed]
table[row sep=crcr]{%
	0.0795378169852995	2.54048395127787e-05\\
	0.159075633970599	0.000101707317145972\\
	0.318151267941198	0.000406854907355269\\
	0.636302535882396	0.00162683233691976\\
	1.27260507176479	0.00657882709165915\\
	2.54521014352958	0.026945389103974\\
	5.09042028705917	0.109113684141927\\
};
\addlegendentry{2S}

\addplot [color=mycolor2, thick, dashed, mark=x, mark options={solid}, forget plot]
table[row sep=crcr]{%
	0.0795378169852995	9.53495615035066e-06\\
	0.159075633970599	4.08656668612672e-05\\
	0.318151267941198	0.00016651545347314\\
	0.636302535882396	0.000671853686945931\\
	1.27260507176479	0.00267278558441025\\
	2.54521014352958	0.0107419047465766\\
	5.09042028705917	0.044925086206681\\
};

\addplot [color=mycolor3, thick, dashdotted]
table[row sep=crcr]{%
	0.0795378169852995	3.65028901434528e-06\\
	0.159075633970599	1.45589522126954e-05\\
	0.318151267941198	5.84219064741376e-05\\
	0.636302535882396	0.000233960353532581\\
	1.27260507176479	0.000936104093927348\\
	2.54521014352958	0.00359459552404946\\
	5.09042028705917	0.0138476927468\\
};
\addlegendentry{1S}

\legend{}
\draw [thick, color=mycolor1] (0.636302535882396, 3e-6) -- (0.636302535882396, 0.2);
\node at (0.85, 0.2) [anchor=north] {(2b)};
\end{axis}
%
%
\definecolor{mycolor1}{rgb}{0.00000,0.44700,0.74100}%
\definecolor{mycolor2}{rgb}{1,0,0}%
\definecolor{mycolor3}{rgb}{0.92900,0.69400,0.12500}%
\definecolor{mycolor4}{rgb}{0.49400,0.18400,0.55600}%
\definecolor{mycolor5}{rgb}{0.46600,0.67400,0.18800}%
%

\begin{axis}[%
name=eig_discr_specamp,
at=(eig_trunc_specamp.south),
anchor=north,
yshift=-\verticaldistance,
width=0.951\figurewidth,
height=\figureheight,
scale only axis,
unbounded coords=jump,
xmode=log,
xmin=0.05,
xmax=10,
xminorticks=true,
xlabel style={yshift=0.5em, font=\color{white!15!black}},
xlabel={Sampling time $T_s\cdot B_{0.999}$},
ymin=0,
ymax=2000,
axis background/.style={fill=white},
legend style={legend cell align=left, align=left, draw=white!15!black}
]
\addplot [color=mycolor1, thick]
table[row sep=crcr]{%
	0.0795378169852995	1956.40790985574\\
	0.159075633970599	1858.71639444459\\
	0.318151267941198	1545.10662462549\\
	0.636302535882396	961.867221931418\\
	1.27260507176479	491.850219104728\\
	2.54521014352958	218.822557204211\\
	5.09042028705917	106.99021713489\\
};
\addlegendentry{$Q_1$}

\addplot [color=mycolor1, thick, forget plot, mark=x]
table[row sep=crcr]{%
	0.0795378169852995	nan\\
	0.159075633970599	nan\\
	0.318151267941198	nan\\
	0.636302535882396	783.89586549089\\
	1.27260507176479	507.263110793415\\
	2.54521014352958	265.305229944966\\
	5.09042028705917	146.090918494646\\
};
\addlegendentry{$Q_2$}

%
%
\legend{}
\draw [thick, color=mycolor1] (axis cs:0.636302535882396, 0) -- (axis cs:0.636302535882396, 2000);
\node at (0.85, 2000) [anchor=north] {(2c)};
\end{axis}
		\setlength{\verticaldistance}{3.2em}
%
%
\definecolor{mycolor1}{rgb}{0.00000,0.44700,0.74100}%
\definecolor{mycolor2}{rgb}{1,0,0}%
\definecolor{mycolor3}{rgb}{0.92900,0.69400,0.12500}%
\definecolor{mycolor4}{rgb}{0.49400,0.18400,0.55600}%
\definecolor{mycolor5}{rgb}{0.46600,0.67400,0.18800}%
%

\begin{axis}[%
name=eig_att_imag,
at=(eig_discr_imag.south),
anchor=north,
yshift=-1.5\verticaldistance,
width=0.951\figurewidth,
height=\figureheight,
scale only axis,
xmin=0,
xmax=1,
xlabel style={yshift=0.5em, font=\color{white!15!black}},
xlabel={$z$},
ymin=0.6,
ymax=2.5,
axis background/.style={fill=white},
title style={yshift=-0.3em, font=\bfseries},
title={Evolution of eigenvalues: $\eta$},
legend style={legend cell align=left, align=left, draw=white!15!black}
]
\addplot [color=mycolor1, thick]
table[row sep=crcr]{%
0	1.25325639203271\\
0.05	1.2228272138639\\
0.1	1.19712306972138\\
0.15	1.17272576890133\\
0.2	1.14954114242583\\
0.25	1.1274145499881\\
0.3	1.10571251127653\\
0.35	1.08556682085875\\
0.4	1.06580949967242\\
0.45	1.04649112560396\\
0.5	1.02767865719424\\
0.55	1.00928184964658\\
0.6	0.991284288152112\\
0.65	0.973956481719281\\
0.7	0.956459820156908\\
0.75	0.939200942658149\\
0.8	0.922133733742717\\
0.85	0.905267027926112\\
0.9	0.88859042708571\\
0.95	0.872133373990975\\
1	0.8558658494076\\
};
\addlegendentry{DS}

\addplot [color=mycolor1, thick, mark=x, forget plot]
table[row sep=crcr]{%
0	1.24667014945073\\
0.05	1.21987743084424\\
0.1	1.18903445513589\\
0.15	1.15801151540196\\
0.2	1.12713668571691\\
0.25	1.09663936037532\\
0.3	1.06655876311271\\
0.35	1.03693960694302\\
0.4	1.00788754690641\\
0.45	0.97936213118092\\
0.5	0.951408431262796\\
0.55	0.924023609244498\\
0.6	0.897293024205684\\
0.65	0.871161310333502\\
0.7	0.845665775856517\\
0.75	0.820888465899107\\
0.8	0.796771459917989\\
0.85	0.773370479453753\\
0.9	0.750644795592549\\
0.95	0.728669010530638\\
1	0.707339596385056\\
};

\addplot [color=mycolor2, thick, dashed]
table[row sep=crcr]{%
0	1.5\\
0.05	1.47489606729475\\
0.1	1.44859037823583\\
0.15	1.42337980352532\\
0.2	1.39846142827325\\
0.25	1.37380881424492\\
0.3	1.3492268861128\\
0.35	1.32443735240186\\
0.4	1.29909099835017\\
0.45	1.27229093913928\\
0.5	1.24360613602755\\
0.55	1.2117717976106\\
0.6	1.17531074465532\\
0.65	1.13333658011287\\
0.7	1.08543450771884\\
0.75	1.03321513041451\\
0.8	0.979233772723303\\
0.85	0.926832338731866\\
0.9	0.87874923876925\\
0.95	0.835765651950124\\
1	0.797773711033971\\
};
\addlegendentry{2S}

\addplot [color=mycolor2, thick, dashed, mark=x, forget plot, mark options={solid}]
table[row sep=crcr]{%
0	1\\
0.05	0.96806750706764\\
0.1	0.935421040681654\\
0.15	0.904270663253503\\
0.2	0.872948983333979\\
0.25	0.842432264392193\\
0.3	0.813007364838174\\
0.35	0.785241855991968\\
0.4	0.759746652839048\\
0.45	0.73711421155038\\
0.5	0.718436328666901\\
0.55	0.704623286808154\\
0.6	0.697059698760278\\
0.65	0.696163886596796\\
0.7	0.701923020551386\\
0.75	0.712475008969268\\
0.8	0.725000004705359\\
0.85	0.736205652497043\\
0.9	0.743724828560379\\
0.95	0.746664374977072\\
1	0.745671580258583\\
};

\addplot [color=mycolor3, thick, dashdotted]
table[row sep=crcr]{%
0	2.5\\
0.05	2.44256611953054\\
0.1	2.38674627213762\\
0.15	2.33193289761154\\
0.2	2.2785331362989\\
0.25	2.22634751853593\\
0.3	2.17503624253243\\
0.35	2.12467070192343\\
0.4	2.07551340633557\\
0.45	2.02747934165018\\
0.5	1.98172495848571\\
0.55	1.9362971961533\\
0.6	1.89194467794821\\
0.65	1.84846781306836\\
0.7	1.80585193261405\\
0.75	1.76422616138453\\
0.8	1.72352495470092\\
0.85	1.68379744854746\\
0.9	1.64497330415128\\
0.95	1.60717043616864\\
1	1.57024176793378\\
};
\addlegendentry{1S}
\legend{}
\node at (0.5, 2.5) [anchor=north] {(3a)};
\end{axis}
%
%
\definecolor{mycolor1}{rgb}{0.00000,0.44700,0.74100}%
\definecolor{mycolor2}{rgb}{1,0,0}%
\definecolor{mycolor3}{rgb}{0.92900,0.69400,0.12500}%
\definecolor{mycolor4}{rgb}{0.49400,0.18400,0.55600}%
\definecolor{mycolor5}{rgb}{0.46600,0.67400,0.18800}%
%

\begin{axis}[%
name=eig_att_normconst,
at=(eig_discr_normconst.south),
anchor=north,
yshift=-1.5\verticaldistance,
width=0.951\figurewidth,
height=\figureheight,
scale only axis,
xmin=0,
xmax=1,
xlabel style={yshift=0.5em, font=\color{white!15!black}},
xlabel={$z$},
ymin=0,
ymax=50,
axis background/.style={fill=white},
title style={yshift=-0.3em, font=\bfseries, align=center},
title={GNFT norming constants: $|Q_{k\ell}|$},
legend style={legend cell align=left, align=left, draw=white!15!black}
]
\addplot [color=mycolor1, thick]
table[row sep=crcr]{%
0	6.25488327781559\\
0.05	6.35613956624051\\
0.1	6.44808250939528\\
0.15	6.5280272547538\\
0.2	6.58969563309598\\
0.25	6.62658904845583\\
0.3	6.63227083458759\\
0.35	6.57641874092866\\
0.4	6.46983252222208\\
0.45	6.30628015210375\\
0.5	6.0889090639168\\
0.55	5.82823410859192\\
0.6	5.53921938121509\\
0.65	5.2220602855529\\
0.7	4.90217111209427\\
0.75	4.58588919947632\\
0.8	4.27898618660571\\
0.85	3.98590410707361\\
0.9	3.7078616715304\\
0.95	3.44579234977694\\
1	3.19405752730195\\
};
\addlegendentry{DS, $|Q_{11}|$}

\addplot [color=mycolor1, thick, mark=x]
table[row sep=crcr]{%
0	42.1624034979029\\
0.05	43.5137053210587\\
0.1	44.7541680356481\\
0.15	45.7972794364241\\
0.2	46.4466547626027\\
0.25	46.5994150330735\\
0.3	46.133088444049\\
0.35	44.527356806492\\
0.4	42.0959237577986\\
0.45	38.796873431358\\
0.5	34.8243796798256\\
0.55	30.4793968398651\\
0.6	26.0847399270574\\
0.65	21.6304101229132\\
0.7	17.7394943080906\\
0.75	14.3301053173312\\
0.8	11.4533114111105\\
0.85	9.13048859366139\\
0.9	7.34844782896379\\
0.95	6.09123915552587\\
1	5.27883629742821\\
};
\addlegendentry{DS, $|Q_{10}|$}

\addplot [color=mycolor2, thick, dashed]
table[row sep=crcr]{%
0	7.2519903715486\\
0.05	7.09610602763532\\
0.1	6.99431271389128\\
0.15	6.88028235567105\\
0.2	6.78476614665554\\
0.25	6.71696774797559\\
0.3	6.68325118399048\\
0.35	6.69688777803552\\
0.4	6.77182087067519\\
0.45	6.9458172804036\\
0.5	7.22082003142374\\
0.55	7.62724799540655\\
0.6	8.27690462620048\\
0.65	9.18597721993591\\
0.7	10.4152122294256\\
0.75	11.9317436614387\\
0.8	13.5430847585933\\
0.85	14.8268313790192\\
0.9	15.4232730678952\\
0.95	15.3331343802482\\
1	14.7901801104328\\
};
\addlegendentry{2S}

\addplot [color=mycolor2, thick, dashed, forget plot, mark=x, mark options={solid}]
table[row sep=crcr]{%
0	20.6873155878469\\
0.05	18.6907636762838\\
0.1	16.6916792629939\\
0.15	15.2042010632559\\
0.2	13.6931874129894\\
0.25	12.254053143775\\
0.3	11.0190477317122\\
0.35	9.85724771730873\\
0.4	8.8201753709334\\
0.45	7.91093257865598\\
0.5	7.07035287897421\\
0.55	6.40282989765383\\
0.6	5.84355187116783\\
0.65	5.44016705212306\\
0.7	5.18253797626875\\
0.75	5.02084778909258\\
0.8	4.85961092943967\\
0.85	4.62855434657783\\
0.9	4.29063535952206\\
0.95	3.9318493788546\\
1	3.62342620005052\\
};

\addplot [color=mycolor3, thick, dashdotted]
table[row sep=crcr]{%
0	5.00472195581252\\
0.05	4.8900398487803\\
0.1	4.77803761327888\\
0.15	4.66812112581361\\
0.2	4.56025213670788\\
0.25	4.45515602416643\\
0.3	4.35282154231699\\
0.35	4.25305239968595\\
0.4	4.15572229873522\\
0.45	4.06032070266327\\
0.5	3.96672915628408\\
0.55	3.87519247844471\\
0.6	3.78597852481752\\
0.65	3.69894464762956\\
0.7	3.61402833371681\\
0.75	3.53132058644909\\
0.8	3.45048794556542\\
0.85	3.37143242856142\\
0.9	3.29401118453139\\
0.95	3.21836299319142\\
1	3.14432225718992\\
};
\addlegendentry{1S}
\legend{}
\node at (0.5, 50) [anchor=north] {(3b)};
\legend{};
\end{axis}
%
%
\definecolor{mycolor1}{rgb}{0.00000,0.44700,0.74100}%
\definecolor{mycolor2}{rgb}{1,0,0}%
\definecolor{mycolor3}{rgb}{0.92900,0.69400,0.12500}%
\definecolor{mycolor4}{rgb}{0.49400,0.18400,0.55600}%
\definecolor{mycolor5}{rgb}{0.46600,0.67400,0.18800}%
%

\begin{axis}[%
name=eig_att_specamp,
at=(eig_discr_specamp.south),
anchor=north,
yshift=-1.5\verticaldistance,
width=0.951\figurewidth,
height=\figureheight,
scale only axis,
xmin=0,
xmax=1,
xlabel style={yshift=0.5em, font=\color{white!15!black}},
xlabel={$z$},
ymin=0,
ymax=1000,
axis background/.style={fill=white},
title style={yshift=-0.3em, font=\bfseries, align=center},
title={DS with NFT processing: $|Q_k|$},
legend style={legend cell align=left, align=left, draw=white!15!black}
]
\addplot [color=mycolor1, thick]
table[row sep=crcr]{%
0	926.305789453992\\
0.05	58.7652895612193\\
0.1	38.8447288726042\\
0.15	29.5644662822543\\
0.2	24.0343615949045\\
0.25	20.1602718390436\\
0.3	17.3316080568498\\
0.35	15.0862865335295\\
0.4	13.267718415405\\
0.45	11.798018502201\\
0.5	10.5788496031566\\
0.55	9.54606929416426\\
0.6	8.66124988062769\\
0.65	7.95487000664525\\
0.7	7.25721174545851\\
0.75	6.65243784346117\\
0.8	6.12055554564769\\
0.85	5.64889035756943\\
0.9	5.2278147105532\\
0.95	4.85017108301876\\
1	4.51060001232834\\
};
\addlegendentry{Double soliton}

\addplot [color=mycolor1, thick, mark=x, forget plot]
table[row sep=crcr]{%
0	932.538792776217\\
0.05	59.9926376085557\\
0.1	40.3659048168064\\
0.15	32.4455351254027\\
0.2	27.3384644149946\\
0.25	23.9510059927706\\
0.3	21.431486347119\\
0.35	19.2720661950249\\
0.4	17.8381611011306\\
0.45	16.4835146851784\\
0.5	15.3238450443406\\
0.55	14.2679400895517\\
0.6	13.4079042698231\\
0.65	12.6089945369136\\
0.7	11.8805322569382\\
0.75	11.1683980253576\\
0.8	10.6041479139642\\
0.85	10.0329918027249\\
0.9	9.49006103693334\\
0.95	9.01689241264779\\
1	8.52494413231182\\
};
\legend{}
\node at (0.5, 1000) [anchor=north] {(3c)};
\end{axis}
%
%
\definecolor{mycolor1}{rgb}{0.00000,0.44700,0.74100}%
\definecolor{mycolor2}{rgb}{1,0,0}%
\definecolor{mycolor3}{rgb}{0.92900,0.69400,0.12500}%
\definecolor{mycolor4}{rgb}{0.49400,0.18400,0.55600}%
\definecolor{mycolor5}{rgb}{0.46600,0.67400,0.18800}%
%

\begin{axis}[%
name=eig_noise_imag,
at=(eig_att_imag.south),
anchor=north,
yshift=-1.3\verticaldistance,
width=0.951\figurewidth,
height=\figureheight,
scale only axis,
xmin=-40,
xmax=2,
xlabel style={yshift=0.5em, font=\color{white!15!black}},
xlabel={$P$ (dBm)},
ymode=log,
ymin=0,
ymax=0.02,
axis background/.style={fill=white},
title style={yshift=-0.3em, font=\bfseries},
title={NMSE of eigenvalues: $\mathrm{NMSE}_{\eta}$},
legend style={legend cell align=left, align=left, draw=white!15!black}
]
\addplot [color=mycolor1, thick]
table[row sep=crcr]{%
	-37.5257498915995	0.00893999976247737\\
	-34.5154499349597	0.00609996502997709\\
	-31.5051499783199	0.00477073362786459\\
	-28.4948500216801	0.00351313380701576\\
	-25.4845500650403	0.0025960092394402\\
	-22.4742501084005	0.00212002174256313\\
	-19.4639501517607	0.00153468844709737\\
	-16.4536501951208	0.00138766002996991\\
	-13.443350238481	0.000987119174059659\\
	-10.4330502818412	0.000894956684691153\\
	-7.42275032520141	0.000690415373502536\\
	-4.4124503685616	0.000449053982777375\\
	-1.40215041192178	0.000326942474060281\\
	1.60814954471803	0.000182617557264895\\
};
\addlegendentry{DS}

\addplot [color=mycolor1, thick, forget plot, mark=x]
table[row sep=crcr]{%
	-37.5257498915995	0.0190374714589276\\
	-34.5154499349597	0.00458892370776392\\
	-31.5051499783199	0.00380529250342982\\
	-28.4948500216801	0.00316343659976254\\
	-25.4845500650403	0.00218332976494035\\
	-22.4742501084005	0.00166386546576914\\
	-19.4639501517607	0.00119077241255453\\
	-16.4536501951208	0.00108447855227524\\
	-13.443350238481	0.000759625956637328\\
	-10.4330502818412	0.000679385998839509\\
	-7.42275032520141	0.000479217108982794\\
	-4.4124503685616	0.000404125395375855\\
	-1.40215041192178	0.000245753206405882\\
	1.60814954471803	0.000188421584978022\\
};

\addplot [color=mycolor2, thick, dashed]
table[row sep=crcr]{%
-37.5257498915995	0.00211030576939723\\
-34.5154499349597	0.00147143541013196\\
-31.5051499783199	0.00108137239282267\\
-28.4948500216801	0.000727982344949873\\
-25.4845500650403	0.000537209140811327\\
-22.4742501084005	0.000356251610404633\\
-19.4639501517607	0.00025608122176492\\
-16.4536501951208	0.00017993420259779\\
-13.443350238481	0.000139706542527728\\
-10.4330502818412	9.64811356411475e-05\\
-7.42275032520141	7.8862510167881e-05\\
-4.4124503685616	0.000118953167409881\\
-1.40215041192178	0.000159997443924081\\
1.60814954471803	0.000113007660273713\\
};
\addlegendentry{2S}

\addplot [color=mycolor2, thick, dashed, mark=x, mark options={solid}, forget plot]
table[row sep=crcr]{%
	-37.5257498915995	0.0106567424904335\\
	-34.5154499349597	0.00668027144924164\\
	-31.5051499783199	0.00493905217466015\\
	-28.4948500216801	0.00309790290715522\\
	-25.4845500650403	0.00229160296669077\\
	-22.4742501084005	0.0017348564786202\\
	-19.4639501517607	0.00102374491750825\\
	-16.4536501951208	0.000840402194556814\\
	-13.443350238481	0.00060960971962807\\
	-10.4330502818412	0.000414667814107552\\
	-7.42275032520141	0.000323727643022708\\
	-4.4124503685616	0.000266293938008685\\
	-1.40215041192178	0.000264137666070666\\
	1.60814954471803	0.000188926648955305\\
};

\addplot [color=mycolor3, thick,  dashdotted]
table[row sep=crcr]{%
	-37.5257498915995	0.00581081715684207\\
	-34.5154499349597	0.00304032614413685\\
	-31.5051499783199	0.00202004652562966\\
	-28.4948500216801	0.00151174517422808\\
	-25.4845500650403	0.000973878150630822\\
	-22.4742501084005	0.000832303724867704\\
	-19.4639501517607	0.000459998198205195\\
	-16.4536501951208	0.000373773778345859\\
	-13.443350238481	0.000254909483032749\\
	-10.4330502818412	0.000193675735014409\\
	-7.42275032520141	0.000129901443374955\\
	-4.4124503685616	9.09562718999039e-05\\
	-1.40215041192178	6.28672138854033e-05\\
	1.60814954471803	4.60228271662845e-05\\
};
\addlegendentry{1-soliton}
\legend{}
\node at (-19, 0.02) [anchor=north] {(4a)};
\end{axis}
%
%
\definecolor{mycolor1}{rgb}{0.00000,0.44700,0.74100}%
\definecolor{mycolor2}{rgb}{1,0,0}%
\definecolor{mycolor3}{rgb}{0.92900,0.69400,0.12500}%
\definecolor{mycolor4}{rgb}{0.49400,0.18400,0.55600}%
\definecolor{mycolor5}{rgb}{0.46600,0.67400,0.18800}%
%

\begin{axis}[%
	name=eig_noise_normconst,
	at=(eig_att_normconst.south),
	anchor=north,
	yshift=-1.3\verticaldistance,
width=0.951\figurewidth,
height=\figureheight,
scale only axis,
xmin=-40,
xmax=2,
xlabel style={yshift=0.5em, font=\color{white!15!black}},
xlabel={$P$ (dBm)},
ymode=log,
ymin=0,
ymax=1,
axis background/.style={fill=white},
title style={yshift=-0.3em, font=\bfseries},
title={GNFT norming constants: $\mathrm{NMSE}_{Q_{k\ell}}$},
legend style={legend cell align=left, align=left, draw=white!15!black}
]
\addplot [color=mycolor1, thick]
table[row sep=crcr]{%
	-37.5257498915995	3.00616848199341\\
	-34.5154499349597	0.945568891827744\\
	-31.5051499783199	1.8170484930236\\
	-28.4948500216801	0.294655782988861\\
	-25.4845500650403	0.155775591282568\\
	-22.4742501084005	0.0936057280940029\\
	-19.4639501517607	0.0527582257443295\\
	-16.4536501951208	0.0409860330696962\\
	-13.443350238481	0.0229864299183212\\
	-10.4330502818412	0.018413003547124\\
	-7.42275032520141	0.0147667870151187\\
	-4.4124503685616	0.016999490199118\\
	-1.40215041192178	0.0374623901214348\\
	1.60814954471803	0.191486223959008\\
};
\addlegendentry{DS, $|Q_{11}|$}

\addplot [color=mycolor1, thick, mark=x]
table[row sep=crcr]{%
	-37.5257498915995	60.3765614421902\\
	-34.5154499349597	12.1224388723525\\
	-31.5051499783199	254.598615215422\\
	-28.4948500216801	1.57193962702624\\
	-25.4845500650403	0.612089283325341\\
	-22.4742501084005	0.244524208271347\\
	-19.4639501517607	0.14430904931054\\
	-16.4536501951208	0.0989997081103174\\
	-13.443350238481	0.0549419185761542\\
	-10.4330502818412	0.0422448630261778\\
	-7.42275032520141	0.0331293144261886\\
	-4.4124503685616	0.0443097956698989\\
	-1.40215041192178	0.131929639841541\\
	1.60814954471803	2.53444692127085\\
};
\addlegendentry{DS, $|Q_{10}|$}

\addplot [color=mycolor2, thick, dashed]
table[row sep=crcr]{%
	-37.5257498915995	0.399362743248117\\
	-34.5154499349597	0.210519269899749\\
	-31.5051499783199	0.142925609205353\\
	-28.4948500216801	0.0837746317080521\\
	-25.4845500650403	0.0559977056701361\\
	-22.4742501084005	0.0375481202913853\\
	-19.4639501517607	0.0244033152076222\\
	-16.4536501951208	0.0187007116327229\\
	-13.443350238481	0.0133055885647555\\
	-10.4330502818412	0.00964854969984156\\
	-7.42275032520141	0.00671851913565407\\
	-4.4124503685616	0.00476163549686458\\
	-1.40215041192178	0.017780084315277\\
	1.60814954471803	0.0801977119608101\\
};
\addlegendentry{2S}

\addplot [color=mycolor2, thick, dashed, forget plot, mark=x, mark options={solid}]
table[row sep=crcr]{%
-37.5257498915995	5.7217176636712\\
-34.5154499349597	1.20070366406848\\
-31.5051499783199	0.521490882442471\\
-28.4948500216801	0.250557064701386\\
-25.4845500650403	0.159060235830116\\
-22.4742501084005	0.110838631055286\\
-19.4639501517607	0.0609535994391578\\
-16.4536501951208	0.0506214273100148\\
-13.443350238481	0.0334900042882992\\
-10.4330502818412	0.0205446603200728\\
-7.42275032520141	0.0146783980057441\\
-4.4124503685616	0.0109505976486489\\
-1.40215041192178	0.0172595061209296\\
1.60814954471803	0.0162620152970983\\
};

\addplot [color=mycolor3, thick, dashdotted]
table[row sep=crcr]{%
	-37.5257498915995	0.0625105576197162\\
	-34.5154499349597	0.0488437803611493\\
	-31.5051499783199	0.0320743663350021\\
	-28.4948500216801	0.0239965327906545\\
	-25.4845500650403	0.0153931535119062\\
	-22.4742501084005	0.012226303898023\\
	-19.4639501517607	0.00733447140519603\\
	-16.4536501951208	0.00622329379308288\\
	-13.443350238481	0.00518848456858229\\
	-10.4330502818412	0.00665103274263549\\
	-7.42275032520141	0.011912097527957\\
	-4.4124503685616	0.0295555810926524\\
	-1.40215041192178	0.0826490207775916\\
	1.60814954471803	0.274054805851085\\
};
\addlegendentry{1S}
\legend{}
\node at (-19, 1) [anchor=north] {(4b)};
\end{axis}
%
%
\definecolor{mycolor1}{rgb}{0.00000,0.44700,0.74100}%
\definecolor{mycolor2}{rgb}{1,0,0}%
\definecolor{mycolor3}{rgb}{0.92900,0.69400,0.12500}%
\definecolor{mycolor4}{rgb}{0.49400,0.18400,0.55600}%
\definecolor{mycolor5}{rgb}{0.46600,0.67400,0.18800}%
%

\begin{axis}[%
name=eig_noise_specamp,
at=(eig_att_specamp.south),
anchor=north,
yshift=-1.3\verticaldistance,
width=0.951\figurewidth,
height=\figureheight,
scale only axis,
xmin=-40,
xmax=2,
xlabel style={yshift=0.5em, font=\color{white!15!black}},
xlabel={$P$ (dBm)},
ymode=log,
ymin=0,
ymax=30,
axis background/.style={fill=white},
title style={yshift=-0.3em, font=\bfseries, align=center},
title={DS with NFT processing: $\mathrm{NMSE}_{Q_k}$},
legend style={at={(0.03, 0.97)}, anchor=north west, legend cell align=left, align=left, draw=white!15!black}
]
\addplot [color=mycolor1, thick]
table[row sep=crcr]{%
-37.5257498915995	18.3690766027877\\
-34.5154499349597	67.3487909900052\\
-31.5051499783199	158.635687758381\\
-28.4948500216801	51398.5130636544\\
-25.4845500650403	154.759790963505\\
-22.4742501084005	62.9055262158215\\
-19.4639501517607	19.3305804962821\\
-16.4536501951208	11.9783881724584\\
-13.443350238481	14.4139390741812\\
-10.4330502818412	10.1698315174344\\
-7.42275032520141	8.16852151479047\\
-4.4124503685616	5.81670637773154\\
-1.40215041192178	5.78099981784064\\
1.60814954471803	5.32434009133158\\
};
\addlegendentry{Double soliton}

\addplot [color=mycolor1, thick, mark=x, forget plot]
table[row sep=crcr]{%
-37.5257498915995	0.403243975598501\\
-34.5154499349597	0.373296515642702\\
-31.5051499783199	0.312403382349301\\
-28.4948500216801	0.376187576591671\\
-25.4845500650403	0.376921599857504\\
-22.4742501084005	0.477518857438795\\
-19.4639501517607	0.455501300765614\\
-16.4536501951208	0.555853707648953\\
-13.443350238481	0.615080179961112\\
-10.4330502818412	0.6531004905073\\
-7.42275032520141	0.739304985031984\\
-4.4124503685616	0.735481798265374\\
-1.40215041192178	0.837526066034518\\
1.60814954471803	1.2063938420498\\
};
\legend{}
\node at (2, 30) [anchor=north east] {(4c)};
\end{axis}
	\end{tikzpicture}
	\caption{First row: effect of truncation on (1a) imaginary part of eigenvalues, (1b) norming constants of a DS, a 2-soliton (2S) and a 1-soliton (1S). The NFT spectral amplitudes of the DS are shown in (1c). Second row: effect of sampling period $T_s$. Third row: effect of attenuation. Fourth row: effect of distributed, non-linearly mixing noise. The legend in the top plot of a column applies to all the plots in that column.}
	\label{fig:implementation}
\end{figure*}
\begin{table}[tbp]\centering
	\caption{Parameters of the pulses used in Fig~\ref{fig:implementation}}
		\label{tab:impl_parameters}
\begin{tabular}{|c|c|c|c|c|}
	\hline
	\textbf{Pulse} & $\lambda$ & $Q_{k\ell}$ & $T_{0.999}$ & $B_{0.999}$ \\
	\hline
	DS & $\lambda_1=1.25j$ & \begin{tabular}{@{}c@{}}$Q_{11}=6.25$ \\ $Q_{10}=40.10$ \end{tabular} & $5.25$ & $12.67$ \\
	\hline
	2S &\begin{tabular}{@{}c@{}}$\lambda_1=1.5j$ \\ $\lambda_2=1j$ \end{tabular} & \begin{tabular}{@{}c@{}}$Q_{1}=20.6956$ \\ $Q_{2}=7.2477$ \end{tabular} & $5.41$ & $12.67$ \\
	\hline
	1S & $\lambda_1=2.5j$ & $Q_{1}=5$ & $1.51$ & $23.75$ \\
	\hline
	
\end{tabular}
\end{table}

We simulated three pulses: a DS, a 2S and a 1S, with the parameters in Table~\ref{tab:impl_parameters}. The three pulses have the same energy, which forces the 1S to have different duration and bandwidth than the 2S and DS. Heuristic optimization was used to obtain a small time-bandwidth product (TBP) for all signals. The duration $T_{0.999}$ of the time interval that contains $99.9\%$ of the signal energy, as well as the bandwidth $B_{0.999}$, are listed in Table~\ref{tab:impl_parameters}. Using a sampling time of $T_s=0.0058 s$, we varied the truncation interval $T$. For a fairer comparison, the results of the simulations are plotted as a function of the ratio $T/T_{0.999}$. For reference, the ratio of truncated energy for the different truncation intervals is plotted in Fig.~\ref{fig:trunc_energy}. The eigenvalues of the generated pulse were found using a Newton-Raphson search algorithm~\cite{mansoor_all}. The first-order spectral amplitudes $Q_1$ and $Q_2$ and the second-order norming constants $Q_{11}$ and $Q_{10}$ of the found eigenvalues were computed using \textit{FBT}: forward-backward computation~\cite{aref_control_detection} with the trapezoidal kernel~\eqref{eq:trap}.
%
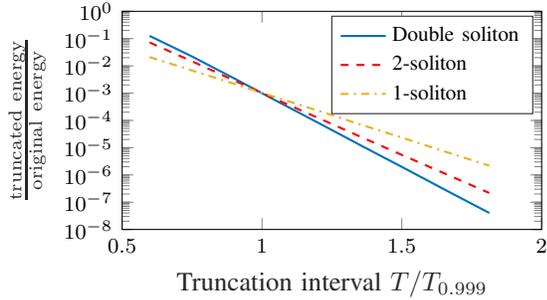
\begin{figure}[tbp]\centering
	\setlength{\figurewidth}{0.32\textwidth}
	\setlength{\figureheight}{0.5\figurewidth}
%
%
\definecolor{mycolor1}{rgb}{0.00000,0.44700,0.74100}%
\definecolor{mycolor2}{rgb}{1,0,0}%
\definecolor{mycolor3}{rgb}{0.92900,0.69400,0.12500}%
\definecolor{mycolor4}{rgb}{0.49400,0.18400,0.55600}%
\definecolor{mycolor5}{rgb}{0.46600,0.67400,0.18800}%
\begin{tikzpicture}

\begin{axis}[%
width=0.961\figurewidth,
height=\figureheight,
at={(0\figurewidth,0\figureheight)},
scale only axis,
xmin=0.5,
xmax=2,
xlabel style={font=\color{white!15!black}},
xlabel={Truncation interval $T/T_{0.999}$},
ymode=log,
ymin=1e-08,
ymax=1,
ytick={1e-8, 1e-7, 1e-6, 1e-5, 1e-4, 1e-3, 1e-2, 1e-1, 1},
yminorticks=true,
ylabel style={font=\color{white!15!black}},
ylabel={$\frac{\mathrm{truncated\ energy}}{\mathrm{original\ energy}}$},
axis background/.style={fill=white},
legend style={legend cell align=left, align=left, draw=white!15!black}
]
\addplot [color=mycolor1, thick]
  table[row sep=crcr]{%
0.598170731707317	0.126038547904397\\
0.74997921286031	0.022108423376954\\
0.901787694013304	0.00342173458872597\\
1.0535961751663	0.000508446697909726\\
1.20540465631929	7.75206632319669e-05\\
1.35721313747228	1.1541367180512e-05\\
1.50902161862528	1.76762365933225e-06\\
1.66083009977827	2.63375449449654e-07\\
1.81263858093126	4.02954120071186e-08\\
};
\addlegendentry{Double soliton}

\addplot [color=mycolor2, thick, dashed]
  table[row sep=crcr]{%
0.598170731707317	0.0723797864341782\\
0.74997921286031	0.0147164482843107\\
0.901787694013304	0.00287467194464797\\
1.0535961751663	0.000573021284324904\\
1.20540465631929	0.000117411901483067\\
1.35721313747228	2.39395193298853e-05\\
1.50902161862528	4.95259814425886e-06\\
1.66083009977827	1.04398384048743e-06\\
1.81263858093126	2.18204752600037e-07\\
};
\addlegendentry{2-soliton}

\addplot [color=mycolor3, thick, dashdotted]
  table[row sep=crcr]{%
0.598170731707317	0.0210178338586561\\
0.74997921286031	0.0067987790562094\\
0.901787694013304	0.00218676901304293\\
1.0535961751663	0.000682383876097092\\
1.20540465631929	0.000219212402154612\\
1.35721313747228	6.83591933572369e-05\\
1.50902161862528	2.19368514275953e-05\\
1.66083009977827	6.84032286668579e-06\\
1.81263858093126	2.19691595559901e-06\\
};
\addlegendentry{1-soliton}

\end{axis}

\end{tikzpicture}%
	\caption{Energy outside truncation interval for the pulses used in Fig~\ref{fig:implementation}}
		\label{fig:trunc_energy}
\end{figure}

Fig.~\ref{fig:implementation}, (1a) shows the relative error $\Delta\eta/\eta$ in the imaginary part of the eigenvalues. For the DS (blue solid curves), truncation splits the eigenvalue into two closely spaced eigenvalues $\lambda_1$ (unmarked curve) and $\lambda_2$ (marked curve), which move closer as $T$ increases. For $T/T_{0.999}\ge 1.357$ (vertical blue solid line), the search algorithm does not distinguish two eigenvalues anymore. Even when two are found, their spectral amplitudes (1c), obtained with NFT processing, are unstable (highly dependent on $T$), because the denominator $a_{\lambda}(\lambda_k)$ is close to $0$. The norming constants in (1b), obtained using the GNFT, are much more stable.

The 2S (red dashed curves) performs better than the DS, probably because the higher order derivatives of $a(\lambda)$ in~\eqref{eq:norming_constants_L2} are less stable. However, a comparison of the blue solid curves in the second and third columns of Fig.~\ref{fig:implementation} shows that GNFT processing is better than NFT when the transmit signal has very closely spaced eigenvalues or higher multiplicity eigenvalues. Note that, for the remainder of the paper, \textit{GNFT processing} is the same as \textit{NFT processing} for a 1S or a 2S. Only for a DS does GNFT processing assume the existence of one eigenvalue with higher multiplicity, while NFT processing assumes two closely spaced eigenvalues. The second column of Figure~\ref{fig:implementation} shows that the 1S (yellow dash-dotted curve) is more robust than the 2S and DS to all impairments. This comes at the cost of spectral efficiency, as the 1S offers fewer degrees of freedom for communication.
\vspace{-0.3em}
\subsection{Effect of pulse discretization}
For the second row of Fig.~\ref{fig:implementation}, we chose a large truncation interval $T=10.5866$ and varied the sampling time $T_s$. Due to the different bandwidths of the pulses, we plot the results as a function of the product $T_s\cdot B_{0.999}$. Again, with enough resolution ($T_s B_{0.999}\le 0.318$), the search algorithm does not find two distinct eigenvalues anymore. The norming constants are much more stable than the spectral amplitudes even when two distinct, closely spaced eigenvalues are found. The error in the norming constants $|\Delta Q_{k\ell}|/|Q_{k\ell}|$ for the DS is almost the same as for the 2S (see plot (2b)).
\vspace{-0.3em}
\subsection{Effect of attenuation}\label{sec:attenuation}
For the third row of Fig.~\ref{fig:implementation}, we simulated the propagation from $z=0$ to $z=1$ of the three pulses with $T_s=0.0058$ and $T=10.5866$ along the noise-free NLSE channel with normalized attenuation coefficient $\alpha=0.4646$ (corresponding to $10$-km propagation and attenuation $0.2\;\mathrm{dB}/\mathrm{km}$ with $\beta_2$ and $\gamma$ from Table~\ref{tab:parameters}). The eigenvalue of the DS splits into two eigenvalues that separate due to attenuation. The attenuation affects the spectral amplitudes (3c) much more strongly than it affects the norming constants (3b). The eigenvalues of the 2S behave similarly to the results in~\cite{prilepsky_breakup}. 

\subsection{Effect of noisy NLSE propagation}
We simulated the propagation of the three pulses along a $4000$-km lossless link with $T_s=0.0771$ and $T=48.43$ and the parameters in Table~\ref{tab:parameters}. The spectral density of the distributed noise was $N_{\mathrm{ASE}}=6.4893\cdot 10^{-24} \;\mathrm{W}\mathrm{s}/\mathrm{m}$. The signal power was varied by changing the free parameter $T_0$. The figure of merit is the normalized mean square error:
\begin{equation}
\mathrm{NMSE}_x=\frac{\mathrm{E}\left[\left|x-\overline{x}\right|^2\right]}{\left|\overline{x}\right|^2}
\end{equation}
where $\overline{x}\triangleq\mathrm{E}[x]$. Again, the instability of the spectral amplitudes is clear from the results in Fig~\ref{fig:implementation}, (4c).

Simulations with different values of the soliton parameters $\lambda_k$, $Q_{k\ell}$ yield similar curves to those in Fig~\ref{fig:implementation} for the three pulses, though the performance comparison between the pulses changes. Only the shape of the attenuation curves (third row) seems to strongly depend on the initial parameters.

\section{Information transmission using the GNFT}\label{sec:experiment}
We simulated a communications system with the parameters in Table~\ref{tab:parameters}. We compared DS, a 2S and a 1S with the same eigenvalues as the pulses in Section~\ref{sec:implementation}. The 1S uses multi-ring modulation on $Q_1=Q_d(\lambda_1)$ with $32$ rings and $128$ phases per ring. The 2S has the two spectral amplitudes $Q_1=Q_d(\lambda_1)$ and $Q_2=Q_d(\lambda_2)$, while the DS has the two norming constants $Q_{11}$ and $Q_{10}$. Both the 2S and the DS have $4$ rings and $16$ phases per spectral amplitude. These parameters were heuristically optimized to obtain a small TBP for all the transmit signals and all positions in $z$. 
The ring amplitudes for the 1S are
\begin{equation}
\left|Q_0\right|\in\left\{0.088754\cdot 1.6142^k \colon k\inset{0}{31}\right\}.
\end{equation}
The ring amplitudes for the 2S and DS are given in Table~\ref{tab:rings_2}. The phases are uniformly spaced in $[0, 2\pi)$, starting at $0$ for $Q_0$, $Q_1$ and $Q_{11}$, and at $\pi/16$ for $Q_2$ and $Q_{10}$. The optimal criterion for choosing ring amplitudes is not known, but expressions such as~\eqref{eq:center_time} suggest that geometric progressions are better suited than arithmetic progressions.

The free parameter $T_0$ in~\eqref{eq:normalize} was used to obtain the desired powers. We used a sampling period of $T_s=0.0771$ and a truncation interval of $T=48.43$ in the simulations.

Lossless propagation according to~\eqref{eq:nlse_analog} was simulated using the split-step Fourier method. In all systems, the transmitter used closed-form expressions to generate the solitons, and the receiver used FBT to obtain the norming constants. Equalization was performed by inverting~\eqref{eq:time_evolution}. The mutual information of the transmitted and received symbols was measured and normalized by the TBP to obtain the spectral efficiency. In the 2S and DS systems, the joint mutual information $I(Q_1^{(TX)}, Q_2^{(TX)}; Q_1^{(RX)}, Q_2^{(RX)})$ was computed, where $Q_k^{(TX)}$ refers to the transmitted symbols and $Q_k^{(RX)}$ refers to the received and equalized symbols. 

Figure~\ref{fig:double_soliton} shows the spectral efficiency for the three systems. At their optimal power, the DS performs better than the 1S, but worse than the 2S. At this point, the DS has broadened in time at most by $14\%$, and the 2S by $11\%$. This small difference is not enough to account for the observed gap in spectral efficiency. The main reason for this gap is the lower stability of the DS: the higher order derivatives in~\eqref{eq:norming_constants_L2} make the norming constants of the DS (especially $Q_{10}$) less stable than the spectral amplitudes of the 2S. The results of Section~\ref{sec:implementation} and Fig.~\ref{fig:implementation} also support this view. However, Fig.~\ref{fig:double_soliton} demonstrates that the generalized NFT with multiple zeros can be used to transmit information. Although the DS does not seem to offer any practical advantage with respect to the 2S, the use of an additional degree of freedom might bring improvements in systems with many eigenvalues, where close spacing is unavoidable.

\begin{table}[t]\centering
	\caption{Simulation parameters}
	\label{tab:parameters}
	\begin{tabular}{|c|c|c|}
		\hline
		\textbf{Parameter} & \textbf{Symbol} & \textbf{Value} \\ 
		\hline
		Dispersion coefficient  & $\beta_2$ & $-21.667\;\mathrm{ps}^2/\mathrm{km}$ \\
		Nonlinear coefficient  & $\gamma$ & $1.2578\;\mathrm{W}^{-1}\mathrm{km}^{-1}$ \\		
		Fiber length & $z$ & $4000\;\mathrm{km}$ \\
		Noise spectral density & $N_{\mathrm{ASE}}$ & $6.4893\cdot 10^{-24} \mathrm{W}\mathrm{s}/\mathrm{m}$ \\		
		\hline		
	\end{tabular}	
\end{table}

\begin{table}[t]\centering
	\caption{Ring amplitudes for the 2S and DS systems}
	\label{tab:rings_2}
	\begin{tabular}{|c|c c c c |}
		\hline
		$|Q_{1}| (\lambda=1.5j)$ & 2.5355 & 2.8364 & 3.1730 & 3.5496 \\ 
		\hline
		$|Q_{2}| (\lambda=1j)$ & 0.2662 & 1.0211 & 3.9173 & 15.0283 \\
		\hline		
%
		$|Q_{11}|$ & 5.3785 & 5.9449 & 6.5708 & 7.2627 \\ 
		\hline
		$|Q_{10}|$ & 34.3750 & 39.3496 & 45.0440 & 51.5625 \\
		\hline		
	\end{tabular}	
\end{table}

\begin{figure}[t]\centering
	\setlength{\figurewidth}{0.75\columnwidth}
	\setlength{\figureheight}{0.50\figurewidth}
%
%
\definecolor{mycolor1}{rgb}{0.00000,0.44700,0.74100}%
\definecolor{mycolor2}{rgb}{1,0,0}%
\definecolor{mycolor3}{rgb}{0.92900,0.69400,0.12500}%
\definecolor{mycolor4}{rgb}{0.49400,0.18400,0.55600}%
\definecolor{mycolor5}{rgb}{0.46600,0.67400,0.18800}%
\begin{tikzpicture}

\begin{axis}[%
width=0.953\figurewidth,
height=\figureheight,
at={(0\figurewidth,0\figureheight)},
scale only axis,
xmin=-40,
xmax=5,
xlabel style={font=\color{white!15!black}},
xlabel={P (dBm)},
ymin=0,
ymax=1.2,
ylabel style={font=\color{white!15!black}},
ylabel={Spectral efficiency (bits/s/Hz)},
axis background/.style={fill=white},
legend style={at={(0.5, 0.03)}, anchor=south, legend cell align=left, align=left, draw=white!15!black}
]

%

\addplot [color=mycolor1, thick]
table[row sep=crcr]{%
	-37.5257498915995	0.0853025290785403\\
	-34.5154499349597	0.26345488137457\\
	-31.5051499783199	0.392352986663292\\
	-28.4948500216801	0.511163032645065\\
	-25.4845500650403	0.605205594658925\\
	-22.4742501084005	0.688105295740116\\
	-19.4639501517607	0.753953292145083\\
	-16.4536501951208	0.813148252803985\\
	-13.443350238481	0.85728904133414\\
	-10.4330502818412	0.879248409334423\\
	-7.42275032520141	0.853877173300721\\
	-4.4124503685616	0.76544488372106\\
	-1.40215041192178	0.546160695877817\\
	1.60814954471803	0.134993585161862\\
};
\addlegendentry{DS}

\addplot [color=mycolor2, thick, dashed]
table[row sep=crcr]{%
	-37.5257498915995	0.259574098038044\\
	-34.5154499349597	0.379790587484555\\
	-31.5051499783199	0.542405255737685\\
	-28.4948500216801	0.670919095287136\\
	-25.4845500650403	0.778962398973691\\
	-22.4742501084005	0.876713692323309\\
	-19.4639501517607	0.96162557569153\\
	-16.4536501951208	1.03198929574809\\
	-13.443350238481	1.08638228646694\\
	-10.4330502818412	1.12274817004579\\
	-7.42275032520141	1.14227119673438\\
	-4.4124503685616	1.14374578722421\\
	-1.40215041192178	1.11662211980483\\
	1.60814954471803	1.01694682775802\\
};
\addlegendentry{2S}

%

\addplot [color=mycolor3, thick, dashdotted]
table[row sep=crcr]{%
	-37.5257498915995	0.482933144801612\\
	-34.5154499349597	0.586232565764654\\
	-31.5051499783199	0.647169544921794\\
	-28.4948500216801	0.674676151262979\\
	-25.4845500650403	0.7044605732511\\
	-22.4742501084005	0.728221661407285\\
	-19.4639501517607	0.735630632443676\\
	-16.4536501951208	0.746087099052493\\
	-13.443350238481	0.742733102607165\\
	-10.4330502818412	0.727604402137002\\
	-7.42275032520141	0.689084009783446\\
	-4.4124503685616	0.646822577648528\\
	-1.40215041192178	0.591277852905414\\
	1.60814954471803	0.505445801026071\\
};
\addlegendentry{1S}

%

\end{axis}
\end{tikzpicture}%
	\caption{Spectral efficiency of three solitonic signals}	
	\label{fig:double_soliton}
\end{figure}
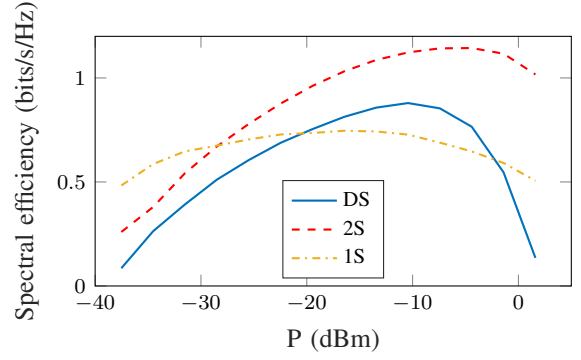


\section{Conclusion}\label{sec:conclusion}
Starting from the theory in~\cite{aktosun_ho, martines_ho}, we proved some properties of the GNFT that are useful for communications. We designed and implemented algorithms to compute the GNFT, and we numerically demonstrated information transmission using higher multiplicity eigenvalues. With this, we extend the class of signals that admit an NFT, providing additional degrees of freedom for NFT-based optical communications. 

There are several directions for future work. Extending the Darboux algorithm to the IGNFT would speed its computation. More insight into the duration, bandwidth and robustness to noise of multiple eigenvalue signals would be useful.

\appendices

\section{Proof of the properties of the GNFT}\label{app:properties}
In the following, all primed variables ($a'$) refer to the spectral functions of the shifted signal $q'(t)$.
\subsubsection{Phase shift} replacing $q$ with $qe^{j\phi_0}$ in~\eqref{eq:zs_ab}, we have $a'(\lambda)=a(\lambda)$ and $b'(\lambda)=b(\lambda)e^{-j\phi_0}$. The property then follows from~\eqref{eq:Q_cont} and~\eqref{eq:Q_kl}.
\subsubsection{Time shift} replacing $t\to t-t_0$ in~\eqref{eq:zs} proves that $a'(\lambda)=a(\lambda)e^{j\lambda t_0}$ and $b'(\lambda)=b(\lambda)e^{-j\lambda t_0}$. The expressions~\eqref{eq:tshift_c} and ~\eqref{eq:tshift_lambda} follow immediately. From~\eqref{eq:Q_kl} we have
\begin{align}
Q_{k\ell}'&=&&\frac{j^\ell}{(L_k-\ell-1)!} \nonumber \\ &&&\cdot\lim\limits_{\lambda\to\lambda_k}\derivk{}{\lambda}{L_k-\ell-1}\left[e^{-2j\lambda t_0}\left(\lambda-\lambda_k\right)^{L_k}\frac{b(\lambda)}{a(\lambda)}\right] \nonumber \\
&=&&e^{-2j\lambda_k t_0}\sum_{u=0}^{L_k-\ell-1}\frac{1}{u!}\left(-2t_0\right)^u Q_{k,\ell+u}
\label{eq:timeshift}
\end{align}
where we applied the product rule
\begin{equation}
\derivk{}{\lambda}{r}\left(f(\lambda)g(\lambda)\right)=\sum_{u=0}^{r}\left(\begin{matrix}r\\u\end{matrix}\right)f^{(u)}(\lambda)g^{(r-u)}(\lambda)
\label{eq:product}
\end{equation}
where
\begin{align*}
f(\lambda)&=\exp(-2j\lambda t_0)\\ g(\lambda)&=(\lambda-\lambda_k)^{L_k}b(\lambda)/a(\lambda).
\end{align*}
The last line of ~\eqref{eq:timeshift} is the same as~\eqref{eq:tshift_d}.
\subsubsection{Frequency shift} using the change of variable $\lambda\to\lambda-\omega_0$ in~\eqref{eq:zs_ab}, we have $a'(\lambda)=a(\lambda-\omega_0)$ and $b'(\lambda)=b(\lambda-\omega_0)$, from which the property follows.
\subsubsection{Time dilation} the change of variable $t\to t/T$ in~\eqref{eq:zs_ab} proves that $a'(\lambda)=a(T\lambda)$ and $b'(\lambda)=b(T\lambda)$. Using this in~\eqref{eq:Q_kl_series} proves~\eqref{eq:tdil}.
\subsubsection{Parseval's theorem} this is a particular case ($n=0$) of the more general \textit{trace formula}:
\begin{align}
C_n=&\frac{1}{\pi}\int_{-\infty}^{\infty}\left(2j\lambda\right)^n\log\left(1+\left|Q_c(\lambda)\right|^2\right)\diff{\lambda} \nonumber \\ & +\frac{4}{n+1}\left(2j\right)^n\sum_{k=0}^{K}L_k\Imag{\lambda_k^{n+1}}
\end{align}
where $C_n$ are the \textit{constants of motion}, of which $C_0$ is the signal energy~\eqref{eq:parseval}. The proof for simple eigenvalues is given in~\cite[Sec. 1.6]{ablowitz_ist}. The result is extended to multiple eigenvalues by allowing several $\zeta_m$ in~\cite[Eq. (1.6.18)]{ablowitz_ist} to be equal.

\section{Proof of Lemma~\ref{th:a_lambda}}\label{app:a_lambda}
Applying Fa\`a di Bruno's formula~\cite[pp. 43-44]{arbogast_faa} to $1/a(\lambda)$, and then the product rule~\eqref{eq:product}, we can write a quotient rule for higher order derivatives:
\begin{multline}
\derivk{}{\lambda}{n}\frac{c}{a}=\sum_{m=0}^{n}\left[\vphantom{\prod_{i=1}^{m}}\binom{n}{m}c^{(n-m)}\right. \\ \left.\sum_{\mathbf{p}\in\mathcal{P}(m)}\frac{(-1)^{|\mathbf{p}|}m!}{p_1!{1!}^{p_1}\cdots {p_m}!{m!}^{p_m}}\frac{|\mathbf{p}|!}{a^{|\mathbf{p}|+1}}\prod_{i=1}^{m}\left(a^{(i)}\right)^{p_i}\right].
\label{eq:quotient}
\end{multline}
Recall that $a^{(i)}$ denotes an $i$-th order derivative. Here, $\mathcal{P}(m)$ denotes the set of partitions $\mathbf{p}$ of $m$:
\begin{equation}
\mathbf{p}=\left[p_1,\cdots,p_m\right], \quad \sum_{i=1}^{m} ip_i=m,\quad p_i\in\mathbb{N}\cup\{0\}
\end{equation}
and $|\mathbf{p}|=\sum_{i=1}^{m} p_i$ is the \textit{cardinality} of $\mathbf{p}$.
Using~\eqref{eq:quotient} in~\eqref{eq:Q_kl} we have
\begin{equation}
Q_{k\ell}=\lim\limits_{\lambda\to\lambda_k}\frac{g(\lambda)}{a(\lambda)^{L_k-\ell}}
\label{eq:t_kl_long}
\end{equation}
where
\begin{multline}
g(\lambda)=j^\ell\sum_{m=0}^{L_k-\ell-1}\left[\vphantom{\prod_{i=1}^{m}}\frac{1}{m!(L_k-\ell-m-1)!}c^{(L_k-\ell-m-1)}(\lambda)\right. \\ \left. \cdot\sum_{\mathbf{p}\in\mathcal{P}(m)}\frac{(-1)^{|\mathbf{p}|}m!}{p_1!{1!}^{p_1}\cdots {p_m}!{m!}^{p_m}}|\mathbf{p}|!a^{L_k-\ell-|\mathbf{p}|-1}\prod_{i=1}^{m}\left(a^{(i)}\right)^{p_i}\right]
\label{eq:g}
\end{multline}
and $c(\lambda)\triangleq (\lambda-\lambda_k)^{L_k}b(\lambda)$. Note that $a$ has a zero of order $L_k$, and therefore $a^{(m)}(\lambda_k)=0$ for $m\inset{0}{L_k-1}$. To compute $Q_{k\ell}$, we repeatedly apply L'H\^opital's rule until the numerator and the denominator become nonzero in the limit:
\begin{equation}
Q_{k\ell}=\frac{g^{(r)}(\lambda_k)}{\left.\left[\mathrm{d}^r a(\lambda)^{L_k-\ell}/\mathrm{d}\lambda^r\right]\right|_{\lambda=\lambda_k}}.
\end{equation}
The number $r$ of times we need to differentiate is the order of the zero in the denominator:
\begin{equation}
r=L_k\left(L_k-\ell\right).
\label{eq:rr}
\end{equation}
The summands in $g^{(r)}(\lambda)$ are of the form
\begin{equation}
g_s(\lambda)=K_s\derivk{}{\lambda}{(L_k-\ell)L_k}c^{(L_k-\ell-m-1)}a^{L_k-\ell-|\mathbf{p}|-1}\prod_{i=1}^{m}\left(a^{(i)}\right)^{p_i}
\label{eq:dg}
\end{equation}
where $s$ is an index, and $K_s$ is a constant independent of $\lambda$. If we apply the product rule to~\eqref{eq:dg}, any nonzero summand at $\lambda=\lambda_k$ must differentiate the factor $c^{(L_k-\ell-m-1)}$ at least $\ell+m+1$ times. Thus, the other factors are differentiated at most $(L_k-\ell)L_k-\ell-m-1$ times. The derivative of
\begin{equation}
a^{L_k-\ell-|\mathbf{p}|-1}\prod_{i=1}^{m}\left(a^{(i)}\right)^{p_i}
\end{equation}
is a sum of terms of the same form. Each new term has the same amount $\sum_i p_i$ of \textit{a-factors} as the original (an \textit{a-factor} here refers to $a$ or one of its derivatives), and the number of differentiations $\sum_{i}ip_i$ in the a-factors is increased by $1$. We conclude that $g_s(\lambda)$ is made up of summands that contain
\begin{itemize}
	\item $L_k-\ell-|\mathbf{p}|-1+\sum_i p_i=L_k-\ell-1$ a-factors that have
	\item $(L_k-\ell)L_k-\ell-m-1+\sum_i ip_i=(L_k-\ell)L_k-\ell-1$ differentiations.
\end{itemize}
All the a-factors of nonzero summands must be at least $L_k$-order derivatives. In the worst case, there are $L_k-\ell-2$ a-factors with an $L_k$-th order derivative. The remaining a-factor must have a derivative of order $\left[(L_k-\ell)L_k-\ell-1\right]-\left[L_k(L_k-\ell-2)\right]=2L_k-\ell-1$.

The product of a-factors in~\eqref{eq:dg} has a zero at $\lambda_k$ of order 
$$L_a=(L_k-\ell-|\mathbf{p}|-1)L_k+\sum_i p_i(L_k-i)=(L_k-\ell-1)L_k-m.$$
This means that any nonzero summand after applying the product rule will differentiate $c^{(L_k-\ell-m-1)}$ at most
$$(L_k-\ell)L_k-L_a=L_k+m$$
times. This yields a term with $c^{(2L_k-\ell-1)}$. As $c=(\lambda-\lambda_k)^{L_k}b$, the  highest order derivative on $b$ is $b^{(L_k-\ell-1)}$.


\section*{Acknowledgment}
The author wishes to thank Prof. G. Kramer and B. Leible for useful comments and proofreading the paper.



\bibliographystyle{IEEEtran}
\bibliography{multiple_eigenvalues}
%
%
%

\end{document}